\def\fullMode{TRUE}
\ifx\fullMode\undefined \newcommand{\FullVer}[1]{}
\newcommand{\ConfVer}[1]{#1} \else \newcommand{\FullVer}[1]{#1}
\newcommand{\ConfVer}[1]{} \fi

\documentclass[10pt,a4paper]{article}
\usepackage{latexsym}
\usepackage[T1]{fontenc}
\usepackage[english]{babel}
\usepackage[nohead,includefoot,margin=2cm]{geometry}
\usepackage[compact,raggedright,small]{titlesec}
\usepackage[affil-it]{authblk}
\usepackage[runin]{abstract}
\usepackage{multicol}
\usepackage{blindtext}
\usepackage{setspace}
\usepackage{amsmath,amsfonts,amsthm,amssymb}
\usepackage{amssymb}
\usepackage{comment}
\usepackage[affil-it]{authblk}
\usepackage{graphicx}
\usepackage{caption}
\usepackage{subcaption}
\usepackage{hhline}
\usepackage{multirow}
\usepackage[stable]{footmisc}
\usepackage{verbatim}
\usepackage[T1]{fontenc}
\usepackage[utf8]{inputenc}
\usepackage{tikz}

\usepackage{balance}
\usepackage{graphicx,xspace} 
\usepackage{fancyhdr}
\usepackage{times,epsfig} 
\usepackage{nccold}
\usepackage{verbatim}
\usepackage[noend]{algorithmic}
\usepackage{color}
\usepackage{url}
\usepackage{float}
\usepackage{alltt}
\usepackage[stable]{footmisc}
\usepackage{amsmath,amsfonts,amsmath,amssymb}
\usepackage{framed}
\usepackage{color}
\usepackage{enumerate}
\usepackage{boxedminipage}
\usepackage{enumerate,paralist}
\usepackage{accents}
\usepackage{algo}
\usepackage[normalem]{ulem}
\usepackage{hyphenat}

\usepackage{times,epsfig} 

\usetikzlibrary{shapes,arrows,fit,shapes.misc}

\newcommand{\ignore}[1]{{}}

\newcommand{\bequal}{\allowbreak=\allowbreak}

\newcommand{\squishlisttwo}{
   \begin{list}{$\bullet$}
       { \setlength{\itemsep}{0pt}      \setlength{\parsep}{3pt}
       \setlength{\topsep}{3pt}       \setlength{\partopsep}{0pt}
      \setlength{\leftmargin}{1.5em} \setlength{\labelwidth}{1em}
       \setlength{\labelsep}{0.5em} } }

\newcommand{\card}[1]{\ensuremath{|#1|}}

\def\eps{\varepsilon}
\def\floor#1{\lfloor {#1} \rfloor}
\def\set#1{\{ {#1} \}}
\def\ceil#1{\lceil {#1} \rceil}
\def\script#1{\mathcal{#1}}
\def\opt{\textsc{OPT}}
\def\sol{\textsc{SOL}}

\def\mA{\script{A}}
\def\mB{\script{B}}
\def\mC{\script{C}}

\def\mF{\script{F}}

\def\mN{\script{N}}

\def\mP{\script{P}}
\def\mQ{\script{Q}}
\def\mR{\script{R}}
\def\mS{\script{S}}
\def\mT{\script{T}}

\def\mX{\script{X}}

\def\NP{\mathbf{NP}}

\def\child{{\tt child}}
\def\root{\mathrm{root}}
\def\node{\mathrm{node}}

\def\part{{\tt part}}
\def\allparts{{\tt all-parts}}
\def\free{{\tt free}}
\def\final{{\tt final}}
\def\leaf{{\tt leaf}}
\def\nullpart{{\tt nullPart}}
\def\scaled{{\ttscl}}
\def\org{{\tt org}}
\def\pr{{\tt pr}}
\def\sol{\mathrm{sol}}
\def\prob#1{\textup{\text{#1}}\xspace}

\addto{\Affilfont}{\small}

\def\Comment#1{\textsl{$\langle\!\langle$#1\/$\rangle\!\rangle$}}

\usepackage{array}
\newcolumntype{P}[1]{>{\centering\arraybackslash}p{#1}}

\newtheorem{theorem}{Theorem}[section]

\newtheorem{example}[theorem]{Example}
\newtheorem{definition}[theorem]{Definition}

\newtheorem{lemma}[theorem]{Lemma}
\newtheorem{problem}[theorem]{Problem}

\title{Cost-Effective Conceptual Design Using Taxonomies}
\author[1]{Ali Vakilian {\it vakilian@mit.edu}}
\author[2]{Yodsawalai Chodpathumwan {\it ychodpa@illinois.edu}}
\author[3]{Arash Termehchy {\it termehca@oregonstate.edu}}
\author[3]{Amir Nayyeri {\it nayyeria@oregonstate.edu}}

\affil[1]{CSAIL, Department of EECS, MIT, Cambridge, MA, USA}
\affil[2]{Department of Computer Science, University of Illinois, Urbana, IL, USA}
\affil[3]{School of EECS, Oregon State University, Corvallis, OR, USA}

\begin{document}

\date{}
\maketitle

\begin{abstract}
It is known that annotating 
named entities in unstructured and semi-structured 
data sets by their concepts improves the effectiveness
of answering queries over these data sets.
Ideally, one would like to annotate entities of 
all concepts in a given domain in a data set,  
however, it takes substantial time and computational 
resources to do so over a large data set.
As every enterprise has a limited budget of time 
or computational resources, it has to annotate 
a subset of concepts in a given domain 
whose costs of annotation do not exceed the budget. 
We call such a subset of concepts a 
{\it conceptual design} for the annotated data set.
We focus on finding a conceptual design that provides 
the most effective answers to queries over the annotated data set, i.e., a 
{\it cost-effective conceptual design}.
Since, it is often less time-consuming and costly 
to annotate small number of general concepts,
such as {\it person}, than 
a large number of specific concepts,  
such as {\it politician} and {\it artist}, we use 
information on superclass/ subclass relationships 
between concepts in taxonomies to find a cost-effective
conceptual design. 
We quantify the amount by which 
a conceptual design with concepts from a taxonomy 
improves the effectiveness of answering queries over an 
annotated data set.
If the taxonomy is a tree, we prove that the problem is
NP-hard and propose an efficient approximation algorithm
and an exact pseudo-polynomial time algorithm for 
the problem. 
We further prove that if the taxonomy is a directed 
acyclic graph, given some generally accepted hypothesis, 
it is not possible to find any approximation algorithm 
with reasonably small approximation ratio or a 
pseudo-polynomial algorithm for the problem. 
Our empirical study using real-world data sets, 
taxonomies, and query workloads shows that our framework 
effectively quantifies the amount by which 
a conceptual design improves 
the effectiveness of answering queries. 
It also indicates that our 
algorithms are efficient for a 
design-time task with pseudo-polynomial algorithm 
being generally more effective than 
the approximation algorithm.
\end{abstract}
\section{Introduction}
\label{sec:introduction}
\subsection{Concept Annotation}
Unstructured and semi-structured data sets,
such as HTML documents,
contain enormous information about 
named entities like people and 
products \cite{ChiticariuLRR10,webconcept:ragu}.
Users normally explore 
these data sets using keyword queries 
to find information about their entities of interest.
Unfortunately, as keyword queries are generally ambiguous,
query interfaces may not return the relevant answers for these
queries. For example, consider the excerpts of the  
Wikipedia\\ ({\it wikipedia.org}) articles in 
Figure~\ref{fig:wikipedia1}. Assume that a user likes 
to find information about {\it John Adams}, the politician,
over this data set. If she 
submits query $Q_1$:{\it John Adams}, 
the query interface may return the articles about 
{\it John Adams}, the artist, or {\it John Adams}, 
the school, as relevant answers. 
Users can further disambiguate their queries 
by adding appropriate keywords. Nonetheless,
it is not easy to find such keywords \cite{YomTov:Difficulty}.
For instance, if one refines $Q_1$ to {\it John Adams Ohio}, 
the query interface may return the article about 
{\it John Adams}, the high school, as the answer.
It will not help either to add keyword {\it Congressman} 
to $Q_1$ as this keyword does not appear in the article 
about {\it John Adams}, the politician. Formulating 
the appropriate keyword query requires some knowledge 
about the sought after entity and the data that 
most users do not usually possess.

\begin{figure}
\centering
\scriptsize{
\begin{alltt}
<article> 
 John Adams has been a former member of the Ohio House of
 Representatives from 2007 to 2014. ...
</article> 
<article> 
 John Adams is a composer whose music is inspired by nature, ...
</article> 
<article> 
 John Adams is a public high school located on the east side of 
 Cleveland, Ohio, ...
</article>
\end{alltt}
}
\vspace{-0.6cm}
\caption{Wikipedia article excerpts}
\label{fig:wikipedia1}
\end{figure}
\begin{figure}
\scriptsize{
\begin{alltt}
<article> 
 {\bf<politician>} John Adams {\bf</politician>} has been a former member 
 of the {\bf<legislature>} Ohio House of Representatives {\bf</legislature>} 
 from 2007 to 2014. ...
</article> 
<article> 
 {\bf<artist>} John Adams {\bf</artist>} is a composer whose music is inspired 
 by nature, ...
</article>
<article>
 {\bf<school>} John Adams {\bf</school>} is a public high school located on 
 the east side of {\bf<city>}Cleveland{\bf</city>}, {\bf<state>}Ohio{\bf</state>}, ...
</article>
\end{alltt} 
}
\vspace{-0.5cm}
\caption{Annotated Wikipedia article excerpts}
\label{fig:wikipedia2}
\end{figure}

To make querying unstructured and semi-structured data sets 
easier, data management researchers have proposed methods to 
identify the mentions to entities 
in these data sets and annotate them by their concepts  \cite{ChiticariuLRR10,webconcept:ragu}. 
Figure~\ref{fig:wikipedia2} shows excerpts of 
the annotated Wikipedia articles whose original versions are 
shown in Figure~\ref{fig:wikipedia1}.  
Because entities in an annotated 
data set are disambiguated by their concepts, 
the query interface can answer queries 
over these data sets more effectively.
Moreover, as the list of concepts used to annotate the data sets
are available to users, they can further clarify their 
queries by mentioning the concepts of
entities in these queries. 
For example, a user who would like to retrieve article(s) 
about {\it John Adams}, the politician, over 
the annotated Wikipedia data set in Figure~\ref{fig:wikipedia2}
may mention the concept of {\it politician} in her query.
The set of annotated concepts in a data set is 
the {\it conceptual design} for the data set \cite{Termehchy:SIGMOD:14}.
For example, the conceptual 
design of the data fragment in Figure~\ref{fig:wikipedia2}
is $D_1$ = \{{\it politician, legislature, artist, school, state, city}\}. Using $D_1$, the query interface is 
able to disambiguate all entities in this data fragment.

\subsection{Costs of Concept Annotation}
Ideally, an enterprise would like to annotate 
all relevant concepts from a data set to answer 
all queries effectively. 
Nonetheless, an enterprise has to spend significant 
time, financial and computational resources, 
and manual labor to accurately extract entities of a concept
in a large data set \cite{Anderson:CIDR:2013,IEMaintenance:Gulhane,ChiticariuLRR10,Termehchy:SIGMOD:14,OptimizeSQLText:Jain,Shen:SIGMOD:08,PrioritizationIE:Huang,Resource:Kanani}. 
An enterprise usually has to develop or obtain a 
complex program called {\it concept annotator} to annotate 
entities of a concept from a collection of documents \cite{McCALLUM:ACMQueue:05}. 
Enterprises develop concept annotator using {\it rule-based}
or {\it machine learning} approaches. 
In the rule-based approach, developers have to design and write 
{\it hand-tuned programming rules} 
to identify and annotate entities of a given concept. 
For example, one rule to annotate  
entities of concept {\it person} is that 
they start with a capital letter. 
It is not uncommon for a rule-based concept annotator to 
have thousands of programming rules, which 
takes a great deal of resources 
to design, write, and debug \cite{McCALLUM:ACMQueue:05}.

One may also use machine learning algorithms 
to develop an extractor for a 
concept \cite{McCALLUM:ACMQueue:05}. In this approach,
developers have to find a set of 
{\it relevant features} for the learning algorithm.
Unfortunately, as the specifications of relevant features are 
usually unclear, developers have to find the 
relevant features through a time-consuming and 
labor-intensive process
\cite{Anderson:CIDR:2013,Anderson:PVLDB:2014}.
First, they have to inspect the data set to
find some candidate features. 
For each candidate feature, developers have to write a
program to extract the value(s) of the feature 
from the data set. Finally,
they have to train and test the concept annotator 
using the set of selected features. 
If the concept annotator is not sufficiently 
accurate, developers have to explore the data set 
for new features. As a concept annotator normally uses 
hundreds of features, developers have to iterate these
steps many times to find a set of reasonably 
effective features, where each iteration usually takes considerable amount of time \cite{Anderson:CIDR:2013,Anderson:PVLDB:2014}.
The overheads feature engineering and computation 
have been well recognized in machine learning community
\cite{Weiss:NIPS:2013}.
Moreover, if concept annotators use supervised learning algorithms, 
developers have to collect or create training data, which 
require additional time and manual labor.

It is more resource-intensive to develop annotators for
concepts in specific domains, such as biology,
as it requires expensive communication between domain 
experts and developers.
Current studies indicate that these communications are not often 
successful and developers have to slog through the data set
to find relevant features for concept annotators in these domains \cite{Anderson:CIDR:2013}. 

Unfortunately, the overheads of developing a concept 
annotator are not one-time costs. Because the structure and content 
of underlying data sets evolve over time, 
annotators should be regularly rewritten and repaired  \cite{IEMaintenance:Gulhane}. 
Recent studies show that many concept annotator 
need to be rewritten in average 
about every two months \cite{IEMaintenance:Gulhane}.
Thus, the enterprise often have to repeat 
the resource-intensive steps of developing a concept annotator 
to maintain an up-to-date annotated data set.

After developing concept annotators, the enterprise
executes them over the data set to generate the annotated 
collection. As most concept annotators perform complex 
text analysis, such as deep natural language parsing, 
it may take them days to process a large data set \cite{OptimizeSQLText:Jain,Shen:SIGMOD:08,PrioritizationIE:Huang,Resource:Kanani}. As the content of the data set evolves, 
extractors should be often rerun to create an updated 
annotated collection.

\subsection{Cost-Effective Conceptual Design}

\begin{figure} 
\centering
\begin{tikzpicture}[->,>=stealth',scale=0.84]
\tikzstyle{every node}=[circle,draw,fill=black!50,inner sep=1pt,minimum width=3pt,font=\scriptsize]
\tikzstyle{every label}=[rectangle,draw=none,fill=none,font=\scriptsize]
\node (CR) [label=right:$athlete$] at (0,0) {};
\node (thing) [label=above right:$thing$] at (1.25,3) {};
\node (agent) [label=left:$agent$] at (0,2) {};
\node (work) [label=below left:$place$] at (4,2) {};
\node (artwork) [label=left:{\it populated place}] at (4.5,1) {};
\node (sculpture) [label=left:$state$] at (4.2,0) {};
\node (painting) [label=right:$city$] at (5.2,0) {};
\node (person) [label=left:$person$] at (-1,1) {};
\node (org) [label=left:$organization\ $] at (2,1) {};
\node (JB) [label=below:$politician$] at (-1.25,0) {};
\node (artist) [label=left:$artist$] at (-3,0) {};
\node (BAFTA) [label=below left:{\it school}] at (1.25,-0.3) {};
\node (club) [label=below right:$legislature$] at (2.5,-0.3) {};
\path [->] (person) edge  (CR); 
\draw [->] (thing) -- (agent);
\draw [->] (agent) -- (person);
\draw [->] (agent) -- (org);
\draw [->] (person) -- (JB);
\draw [->] (person) -- (artist);
\draw [->] (org) -- (BAFTA); 
\draw [->] (org) -- (club); 
\draw [->] (thing) -- (work); 
\draw [->] (work) -- (artwork);
\draw [->] (artwork) -- (sculpture); 
\draw [->] (artwork) -- (painting); 
\end{tikzpicture}
\caption{Fragments of DBpedia taxonomy from {\it dbpedia.org}}
\label{fig:DBpedia}
\end{figure}
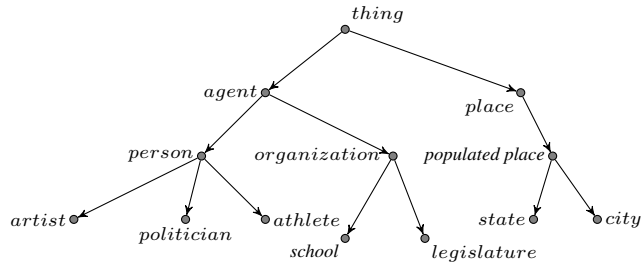
\begin{figure}
\scriptsize{
\begin{alltt}
<article> 
 {\bf<person>} John Adams {\bf</person>} has been a former member 
 of the {\bf<organization>} Ohio House of Representatives {\bf</organization>} 
 from 2007 to 2014. ...
</article> 
<article> 
 {\bf<person>} John Adams {\bf</person>} is a composer whose music is inspired 
 by nature, ...
</article>
<article>
 {\bf<organization>} John Adams {\bf</organization>} is a public high school 
 located on the east side of {\bf<city>}Cleveland{\bf</city>}, {\bf<state>}Ohio{\bf</state>}, 
 ...
</article>
\end{alltt} 
}
\vspace{-0.5cm}
\caption{Wikipedia article excerpts organized in more general concepts}
\label{fig:wikipedia-general}
\end{figure}

Because the available financial or computational 
resources of an enterprise are limited, it may not afford 
to develop, deploy, and maintain annotators for all concepts 
in a domain. Also, many users 
may need an annotated data set quickly 
and cannot wait days for an (updated) annotated collection
\cite{Shen:SIGMOD:08,OptimizeSQLText:Jain}. 
For example, a reporter who pursues some breaking news, 
a stock broker that studies the relevant news and documents
about companies, and
an epidemiologist that follows the 
pattern of a new potential pandemic on the Web
and social media need relevant answers to 
their queries fast.
Hence, the enterprise may afford to annotate only a subset of 
concepts in a domain. 

Concepts in many domains are organized 
in taxonomies \cite{WebDataManage:Abiteboule:11}.
Figure~\ref{fig:DBpedia} depicts 
fragments of DBPedia {\it dbpedia.org} taxonomy, where nodes
are concepts and edges show superclass/ subclass relationships.
An enterprise can use the information in a taxonomy to find 
a conceptual design whose associated costs do not exceed 
its budget and deliver reasonably effective answers for 
queries. For example, assume that 
because an enterprise has to develop in-house annotators for
concepts {\it politician} and {\it artist},   
the total cost of 
annotating concepts in conceptual design 
$D_1$ = 
\{{\it politician, artist, legislature, school, state, city}\}
over original Wikipedia 
collection exceeds its budget.
As some free and reasonably accurate annotators 
are available for concept {\it person}, \\
e.g. {\it nlp.stanford.edu/software/CRF-NER.shtml}, 
the enterprise may annotate concept {\it person}
using smaller amount of resources than concepts 
{\it politician} and {\it artist}. 
Hence, it may afford to annotate concepts 
$D_2$ =  
\{{\it person, organization, state, city}\}
from this collection. 
Thus, the enterprise may choose to annotate the data set 
using $D_2$ instead of $D_1$.
Figure~\ref{fig:wikipedia-general} demonstrates 
the annotated version of the excerpts of Wikipedia articles 
in Figure~\ref{fig:wikipedia1} using conceptual design $D_2$.

Intuitively, a query interface can disambiguate fewer 
queries over the data fragment in Figure~\ref{fig:wikipedia-general} 
than the one in Figure~\ref{fig:wikipedia2}.
For instance, if a users ask for 
information about {\it John Adams}, the politician,
over Figure~\ref{fig:wikipedia-general}, 
the query interface may return the document that 
contains information about {\it John Adams}, the artist, 
as an answer as both entities are annotated as {\it person}. 
Nonetheless, the annotated data set in Figure~\ref{fig:wikipedia-general}
can still help the query interface to disambiguate some queries.
For example, the query interface can recognize the 
occurrence of entity {\it John Adams}, the school, 
from the people named {\it John Adams} in 
Figure~\ref{fig:wikipedia-general}. Thus, it can 
answer queries about the school entity over this data fragment 
effectively. 
Clearly, an enterprise would like to select a 
conceptual design whose required time and/or resources 
for extraction do not exceed its budget 
and most improves the effectiveness of answering queries.
We call such a conceptual design for an annotated data 
set, a {\it cost-effective conceptual design} for the data set.

\subsection{Our Contributions}
Currently, concept annotation experts use their intuitions 
to discover cost-effective conceptual designs 
from taxonomies. Because most taxonomies
contain hundreds of concepts \cite{Suchanek:YAGO}, 
this approach does not scale for real-world applications.
In this paper, we introduce and formalize the 
problem of finding cost-effective conceptual designs 
from taxonomies and
propose algorithms to solve the problem in general and interesting
special cases.
To this end, we make the following contributions.
\squishlisttwo
\item We develop a theoretical framework 
that quantifies the amount of improvement in 
the effectiveness of answering queries by annotating 
a subset of concepts from a taxonomy. 
Our framework takes into account possibility of error in concept
annotation.

\item We introduce and formally define the problem 
of cost-effective
conceptual design over tree-shaped taxonomies and show
it to be NP-hard. 

\item We propose an efficient approximation algorithm, called
the level-wise algorithm, and prove that it has
a bounded worst-case approximation ratio in an interesting 
special case of the problem. We also propose an 
exact algorithm for the problem with pseudo polynomial 
running time.

\item We further define the problem over taxonomies that 
are directed acyclic graphs and 
prove that given a generally accepted hypothesis, 
there is no approximation algorithm 
with reasonably small approximation ratio and
no algorithm with pseudo polynomial running time for this problem.
We show that these results hold even for some restricted cases
of the problem, such as the case where all concepts are equally costly.

\item We evaluate the accuracy of our formal framework
using a large scale real-world data set, Wikipedia, 
real-world taxonomies \cite{Suchanek:YAGO}, 
and a sample of a real-world query workload.
Our results indicate that the formal framework 
accurately measures the amount of improvement 
in the effectiveness of answering queries 
using a subset of concepts from a taxonomy.

\item We perform extensive empirical studies 
to evaluate the accuracy and efficiency of the proposed 
algorithms over real-world data sets, taxonomies, and query workload. 
Our results indicate that the pseudo polynomial algorithm 
is generally able to deliver more effective schemas 
that the level-wise algorithm in reasonable amounts of time. 
They further show that level-wise algorithm provides more effective 
conceptual designs than the pseudo polynomial algorithm if the distribution of concepts in queries is skewed.
\end{list}

The paper is organized as follows. Section~\ref{sec:background} reviews the related work. 
Section~\ref{sec:cost-effective-design} 
formalizes the problem of cost-effective 
conceptual design over a tree-shaped taxonomy
and show that it is NP-hard. 
Section~\ref{sec:approximation-algorithms} describes an efficient
approximation algorithm with 
bounded approximation ratio in an interesting special case of the problem. 
Section~\ref{sec:pseudo-polynomial}
proposes a pseudo-polynomial algorithm for the problem in 
general case.
Section~\ref{sec:dag-taxonomy} defines the problem over 
taxonomies that are directed acyclic graphs and 
provides interesting hardness results for this setting. 
Section~\ref{sec:conclusion} concludes the paper.
The proofs for the theorems of the paper are in the appendix.
\section{Related Work}
\label{sec:background}
Researchers have noticed the overheads and costs
of curating and organizing large data sets \cite{Dong:2012:LMS:2448936.2448938,Resource:Kanani,OptimizeSQLText:Jain}. 
For example, some researchers have recently considered 
the problem of selecting data sources for fusion such that 
the marginal cost of acquiring a new data source 
does not exceed its marginal gain, where cost and gain are 
measured using the same metric, e.g., US dollars 
\cite{Dong:2012:LMS:2448936.2448938}. 
Our work extends this line of research by 
finding cost-effective designs over unstructured 
or semi-structured data sets, which help users query
explore these data sets more easily. 
We also use a different model, where the cost and 
benefit of annotating concepts can be measured 
in different units.

There is a large body of work on building 
large-scale data management systems for 
annotating and extracting entities and relationships 
from unstructured and semi-structured data sources 
 \cite{ChiticariuLRR10,webconcept:ragu}.
In particular, researchers have proposed several techniques 
to optimize the running time, 
required computational power, and/or storage consumption 
of concept annotation programs by processing only 
a subset of the underlying collection that is 
more likely to contain mentions to entities 
of a given concept
\cite{SearchCrawl:Ipeirotis,OptimizeSQLText:Jain,PrioritizationIE:Huang,Resource:Kanani}.
Our work complements these efforts by 
finding a cost-effective set of concepts for annotation 
in the design phase. 
Further, our framework can handle 
other types of costs in creating and maintaining 
annotated data set other than computational overheads.

Researchers have examined the problem of 
selecting a cost effective subset of concepts 
from a set of concepts for annotation \cite{Termehchy:SIGMOD:14}.
Concepts in many real-world domains, however, are maintained
in taxonomies rather than unorganized sets.
We build on this line of work by 
considering the superclass/ subclass relationships 
between concepts in taxonomies to find 
cost-effective designs. 
Because taxonomies have richer structures 
than sets of concepts, they present new 
opportunities for finding cost-effective designs. 
For instance, an enterprise may not have 
sufficient budget to annotate a concept $C$ in 
a dataset, but have adequate resources to 
annotate occurrences of a superclass of $C$, such as $D$, 
in the dataset. Hence, 
to answer queries about entities of $C$, 
the query interface may examine only 
the documents that contain mentions to the entities of $D$.
As the query interface does not need to 
consider all documents in the data set,
it is more likely that it returns 
relevant answers for queries about $C$.
Because the algorithms proposed in \cite{Termehchy:SIGMOD:14} 
do not consider superclass/ subclass relationships between 
concepts, one cannot use them to find cos-effective
designs over taxonomies. 
Moreover, as we prove in this paper, it is more 
challenging and harder to find cost-effective designs
over taxonomies than over sets of concepts. 

Researchers have proposed methods to 
semi-automatically construct or expand taxonomies by 
discovering new concepts from large text collections \cite{Clarka:IPM:2012}.
We, however, focus on the problem of annotating instances 
of the concepts in a given taxonomy over an unstructured
or semi-structured data set.

Conceptual design has been an important problem 
in data management from its early days \cite{DBBook}. 
Generally, conceptual designs have been created manually 
by experts who identify the relevant concepts in 
a domain of interest. 
Because an enterprise may not afford to 
annotate the instances of all relevant concepts
in a domain, this approach cannot be applied to 
large-scale concept annotation. As a matter of fact, 
our empirical studies indicate that 
adapting this approach does not generally return 
cost-effective conceptual designs for annotation.
Researchers have studied the problem of 
predicting the costs of developing or maintaining 
pieces of software \cite{Boehm:SoftwareCost}.
Our work is orthogonal to the methods used for estimating
the costs of creating and maintaining concept 
annotation modules.

\section{Cost-Effective Conceptual Design}
\label{sec:cost-effective-design}

\subsection{Basic Definitions}
Similar to previous works, we do not rigorously define the notion 
of named entity \cite{WebDataManage:Abiteboule:11}. 
We define a named entity (entity for short) as a unique name in some (possibly infinite) domain.
A concept is a set of entities, i.e., its instances. 
Some examples of concepts are {\it person} and {\it country}. An entity of concept 
{\it person} is {\it Albert Einstein} and an entity of concept {\it country} is {\it Jordan}.
Concept $C$ is a {\it subclass} of concept 
$D$ iff we have $C \subset D$.
In this case, we call $D$ a {\it superclass} of $C$. 
For example, {\it person} is a superclass of {\it scientist}.
If an entity belongs to a concept C, it will belong to all 
its superclass's.

A taxonomy organizes concepts in a domain of interest \cite{WebDataManage:Abiteboule:11}.
We first investigate the properties of tree-shaped 
taxonomies and later in Section~\ref{sec:dag-taxonomy}
we will explore the taxonomies that are 
directed acyclic graphs.
Formally, we define {\it taxonomy} $\mX = $ $(R, \mC, \mR)$  
as a rooted tree, with {\it root concept} $R$, vertex set $\mC$ and edge set $\mR$.
$\mC$ is a finite set of concepts.  
For $C,D\in \mC$ we have $(C,D) \in \mR$ iff 
$D$ is a subclass of $C$.
Every concept in $\mC$ that is not a superclass of any other 
concept in $\mC$ is a {\it leaf concept}.
The leaf concepts are leaf nodes in taxonomy $\mX$. 
For instance, concepts {\it athlete} and {\it artist} are leaf concepts in Figure~\ref{fig:DBpedia}.
Let $ch(C)$ denote the children of concept $C$.
For the sake of simplicity, we assume that 
$\cup_{D \in ch(C)} D$ $= C$ for all concepts $C$ in a 
taxonomy.

Each data set is a set of documents. 
Data set $DS$ is in the domain of taxonomy $\mX$ iff 
some entities of concepts in $\mX$ appear in
some documents in $DS$.
For instance, the set of documents in 
Figure~\ref{fig:wikipedia1} are in the domain of
the taxonomy shown in Figure~\ref{fig:DBpedia}.
An entity in $\mX$ may appear in
several documents in a data set.  
For brevity, we refer to the occurrences of entities 
of a concept in a data set as the occurrences of 
the concept in the data set.

A query $q$ over $DS$ is a pair $(C, T)$, 
where $C \in \mC$ and $T$ is a set of terms. 
Some example queries are 
{\it (person, \{Michael Jordan\})} or  {\it (location, \{Jordan\})}.  
This type of queries has been widely used to 
search and explore annotated data sets
\cite{OptimizingAnnotation:Chakrabarti,Chu-Carroll06,KeywordStructuredQuery:Pound}. 
Empirical studies on real world query logs indicate that the majority of 
entity centric queries refer to a single entity \cite{MoreSenses:Sanderson}. 
In this paper, we consider queries that refer to a single entity.
Considering more complex queries that seek information 
about relationships between several entities requires
more sophisticated models and algorithms and more space than a paper.
It is also an interesting topic for future work.  

\subsection{Conceptual Design}
{\it Conceptual design} $\mS$ over taxonomy $\mX = $ $(R, \mC, \mR)$ is a 
non-empty subset of $\mC - \set{R}$. For brevity, in the rest of the paper,  
we refer to conceptual design as {\it design}.
A design divides the set of leaf nodes in $\mC$ into some partitions,
which are defined as follows.

\begin{definition}
\label{def:partition}
Let $\mS$ be a design over taxonomy $\mX =$ $(\mC, \mR)$, and let $C\in\mS$.
We define the partition of $C$ as a subset of leaf nodes of 
$\mC$ with the following property.
A leaf node $D$ is in the partition of $C$ iff $D=C$ or $C$ is the lowest ancestor of $D$ in $\mS$.
\end{definition}
\noindent
Let function $\part$ map each concept into its partition.  

\begin{example} 
Consider the taxonomy described in Figure~\ref{fig:annotation}. Let
design $\mS$ be $\{agent, person\}$. 
The partitions of $\mS$ are $\{artist,$ $politician,$ $athlete\}$ 
and $\set{school, legislature}$. 
Also, $\part(\text{person})$ = $\{artist,$ $politician,$ $athlete\}$  
and $\part(\text{agent})$ = $\{school,$
$legislature\}$.
\end{example}  
For each design $\mS$, the set of \emph{leaf concepts} that do not belong to any partition are 
called {\it free concepts} and denoted as $\free(\mS)$.
These concepts neither belong to $\mS$ nor are descendant 
of a concept in $\mS$.
\begin{figure}[htb]
  \centering
    \includegraphics[height=1.75in,width=2.25in]{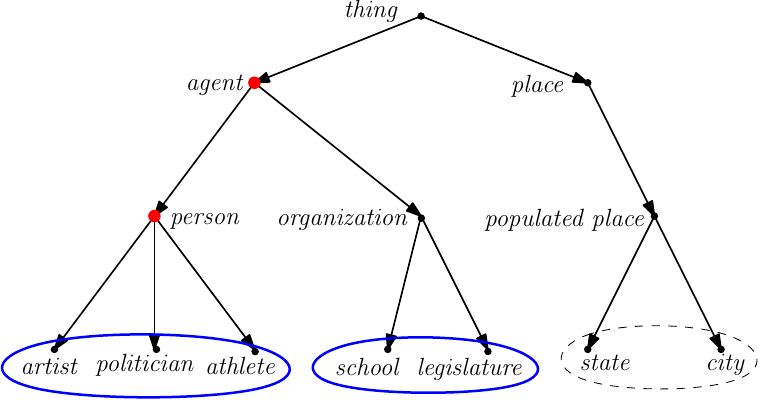}
      \caption{The concepts in red, {\it agent} and {\it person}, denote the design. The blue curves denote the partitions created after annotating the design and the dashed curved shows the free concepts of the selected design.}
  \label{fig:annotation}
\end{figure}

\begin{example} 
Again consider design $\{person,$ $agent\}$ 
over the taxonomy described in Figure~\ref{fig:annotation}. 
The free concepts of $\mS$ are $\{state,$ $city\}$ as 
they are not in any partition of $\mS$. 
\end{example} 

Let $DS$ be a data set in the domain of taxonomy
$\mX =$$(R, \mC, \mR)$ and $\mS$ be a design over $\mX$.
$\mS$ is the design of data set $DS$ iff for all concept $C\in \mS$,  
all occurrences of concepts in the partition of $C$ are annotated by $C$.
In this case, we say $DS$ is an {\it instance} of $\mS$.
For example, consider the design 
$\mT = $ \{{\it person, organization}\}
over the taxonomy in Figure~\ref{fig:DBpedia}.
The data set in Figure~\ref{fig:wikipedia-general}
is an instance of $\mT$  as all instances of concepts 
{\it athlete}, {\it artist} and {\it politician}, that belong 
to the partition of {\it person}, are annotated by 
{\it person} and all instances of concepts
{\it school} and {\it legislature}, that constitute 
the partition of {\it organization},
are annotated by {\it organization} in the data set.

\subsection{Design Queriability}
\label{sec:design-queriability}
Let $\mQ$ be a set of queries over data set $DS$. 
Given design $\mS$ over taxonomy 
$\mX =$$(R, \mC, \mR)$, we would like to measure the
degree by which $\mS$ improves the effectiveness of 
answering queries in $\mQ$ over $DS$. 
The value of this function should be larger for the 
designs that help the query interface to answer a 
larger number of queries in $\mQ$ more effectively.
As most entity-centric information needs are precision-oriented \cite{Classification-Enhanced:Bennett,Chu-Carroll06},  
we use the standard metric of 
{\it precision at $k$} ($p@k$ for short)
to measure the effectiveness of answering 
queries over structured data sets \cite{IRBook}.
The value of $p@k$ is
the fraction of relevant answers in the 
top $k$ returned answers for the query.
We average the values of $p@k$ 
over queries in $\mQ$ to measure the amount of
effectiveness in answering queries in $\mQ$.
The problem of design in order to maximize 
other objective functions, such as recall, 
is an interesting subject for future work.

Let $Q:(C,T)$ be a query in $\mQ$ such that
$C$ belongs to the partition of $P \in \mS$.
The query interface may consider only the documents that 
contain information about entities annotated by 
$P$ to answer $Q$. 
For instance, consider query $Q_1 =$ $(politician,$ $John Adams)$
over data set fragment in Figure~\ref{fig:wikipedia-general}
whose design is \{{\it person, organization}\}. The query interface
may examine only the entities annotated by {\it person}
in this data set to answer $Q_1$.
Thus, the query interface will avoid 
non-relevant results that otherwise may have been placed in 
the top $k$ answers for $Q$. 
It may further rank them according to 
its ranking function, such as the traditional TF-IDF 
scoring methods \cite{IRBook}. 
Our model is orthogonal to the method used to rank the 
candidate answers for the query. 

The query interface still has to examine all documents that 
contain some mentions to the entities annotated by
concept $P$ to answer $Q:(C,T)$. 
Nevertheless, only a fraction
of these documents may contain 
information about entities of $C$. For instance,
to answer query $(politician, John Adams)$ 
over the data set fragment in 
Figure~\ref{fig:wikipedia-general}, the query interface has to examine all documents that contain instances 
of concept {\it person}.
Some documents in this set have matching 
entities form concepts other 
than {\it politician}, such as John Adams, the artist.
We like to estimate the fraction of the results for $Q:(C,T)$ 
that contains a matching entity in concept $C$. Given all other
conditions are the same, the larger this fraction is, the more 
likely it is that the query interface delivers more relevant answers,
and therefore, a larger value of $p@k$ for $Q$.

Let $d_{DS}(C)$ denote the fraction of 
documents that contain entities of concept $C$ in data set $DS$. 
We call $d_{DS}(C)$ the {\it frequency} of $C$ over $DS$. 
When $DS$ is clear from the context, 
we denote the frequency of $C$ as $d(C)$. 
We want to compute the fraction of the returned 
answers for query $Q:(C,T)$ that 
contain a matching instance of concept $C$. 
These entities are annotated by concept 
$P$, such that $C$ is in the partition of $p$.
Let $d(P)$ be the total frequency of leaf concepts
in the partition of $P$. The fraction of these documents
 that contain information about $C$ is $\frac{d(C)}{d(P)}$.
The larger this fraction is, the 
more likely it is that query interface returns more
documents about entities of concept $C$
for query $Q:(C,T)$. Thus, it is more
likely for query interface to return relevant answers
for $Q$ and improve its $p@k$.
For instance, assume that 
the mentions to the entities of concept {\it artist} 
appear more frequently 
in data set $DS$ than the ones of concept {\it politician}. 
Also assume that we only annotate {\it person} from $DS$. 
Given query $(politician, John Adams)$ 
it is more likely for articles about John Adams, the artist,
to appear in the top-ranked answers than about John Adams, the politician. 

We call the fraction of queries in 
$\mQ$ whose concept is $C$ the {\it popularity} of $C$ in 
$\mQ$. Let $u_{\mQ}$ be the function that maps concept 
$C$ to its popularity in $\mQ$. 
When $\mQ$ is clear from the context, we simply use $u$ instead of $u_{\mQ}$. The degree of improvement in value of $p@k$
in answering queries of concept $C$ over $DS$ 
is proportional to $\frac{u(C)\ d(C)}{d(p)}$.
Hence, the amount of the contribution of queries of the concepts 
in partition of $P$ to the value of $p@k$ will be:
$$\sum_{C \in \part(P)} \frac{u(C)\ d(C)}{d(P)}.$$ 
Given all other conditions are the same, the larger this value 
is, the more likely it is that the query interface will achieve a larger $p@k$ value over queries in $\mQ$.

Annotators, however, may make mistakes in identifying 
the correct concepts of entities in a collection \cite{Chu-Carroll06}. 
An annotator may recognize some appearances of entities 
from concepts that are not $P$ as the 
occurrences of entities in $P$.
For instance, the annotator of concept {\it person}
may identify {\it Lincoln}, the movie, as a person. 
The {\it accuracy} of annotating concept $P$ over $DS$ is
the number of correct annotations of $P$ divided by 
the number of all annotations of $P$ in $DS$. 
We denote the accuracy of annotating concept $P$
over $DS$ as $\pr_{DS}(P)$. When $DS$ is clear from the 
context, we show $\pr_{DS}(P)$ as $\pr(P)$. 
Hence, we refine our estimate to the following.
\begin{equation}
\sum_{C \in \part(P)}\frac{u(C)\ d(C)}{d(P)}\ \pr(P).
\label{eq:queriability2}
\end{equation}
\noindent  

Next, we compute the amount of improvement 
that $\mS$ provides for queries whose concepts do not belong
to any partition, i.e., free concepts.
If concept $C$ is a free concept with regard to design
$\mS$, the query interface has to examine all documents in the 
collection to answer $Q:(C,T)$. 
Thus, if $C$ is a free concept, 
the fraction of returned answers 
for $Q$ that contains a matching instance of 
concepts $C$ is $d(C)$.
Using equation~\ref{eq:queriability2}, we formally
define the function that estimates the likelihood of 
improvement for the
value of $p@k$ for all queries in a query workload 
over a data set annotated by design $\mS$. 
\begin{definition}
\label{def:queriability}
The queriability of design $\mS$ from taxonomy $\mX$ 
over data set $DS$ is
\begin{equation}
QU(\mS) = \sum_{P \in \mS} {\sum_{C\in \part(P)} \frac{u(C)\ d(C)\ \pr(P)}{d(P)}}\ + {\sum_{C\in \free(\mS)} u(C)d(C)}.
\label{eq:queriability}
\end{equation}
\noindent  
\end{definition}

Similar to other optimization problems in data management, such
as query optimization \cite{DBBook}, 
the complete information about the
parameters of the objective function, i.e. 
frequencies and popularities of concepts, 
may not be available at the design-time.
Nevertheless, our empirical results in 
Section~\ref{sec:experiment} indicate that one 
can effectively estimate these parameters 
using a small sample of the full data set. 
For instance, we show that 
the frequencies of concepts over a collection of more than
a million documents can be effectively estimated using a 
sample of about three hundred documents.

\subsection{Cost-Effective Design Problem}
Given taxonomy $\mX=$ $(\mC, \mR)$ and data set $DS$
in domain of $\mX$, the 
function $w_{DS}: \mC \rightarrow \mathbb{R}^+$, maps each concept 
$C$ to a real number that reflects the amount of 
resources used to annotate mentions of entities in 
$C$ from data set $DS$. When the data set is clear from the 
context, we simply denote the cost function as $w$. 
The enterprise may predict the costs of 
development and maintenance of annotation programs using 
available methods for predicting costs of software 
development and maintenance \cite{Boehm:SoftwareCost}.
If the cost is running time, the enterprise may use
current methods of estimating the 
execution time of concept annotators \cite{OptimizeSQLText:Jain}.
If there is not sufficient information to estimate the costs for concepts, 
the enterprise may assume that all concepts are equally costly.
We will show in Sections~\ref{sec:approximation-algorithms},
\ref{sec:pseudo-polynomial}, and \ref{sec:dag-taxonomy} that 
finding cost-effective designs is still challenging 
in the cases where concepts are equally costly. 

Similar to previous works on cost-effective concept 
annotation \cite{Termehchy:SIGMOD:14}, 
we assume that annotating certain concepts 
does not affect the cost and accuracies of other concepts.
The reasons behind this assumption are two-fold.
First, it usually takes significant amount of resources to
develop, execute, and maintain a concept annotator
even after pairing with other annotators. 
For instance, developers have to discover 
a large number of distinct features for each concept
to accurately annotate them.
Second, it may require exponential 
number of cost values to express the relationships 
between costs of concepts in a taxonomy,
which is not realistic and makes the problem 
extremely complex to express.
However, finding a simplified framework that 
can effectively express the problem with 
relationships between the costs of annotating different concepts  
is an interesting subject for future work. 

The cost of annotating a data set under
design $\mS$ is the sum of the costs the concepts in $\mS$.
Budget $B$ is a positive real number that 
represents the amount of available 
resources for organizing the data set.
Next, we formally define the problem of Cost-Effective 
Conceptual Design ({\it CECD} for short) as follows.
\begin{problem}
Given taxonomy $\mX$, data set $DS$ in the domain of $\mX$, 
and budget $B$, we like to find design $\mS$ over $\mX$
such that $\sum_{C \in \mS} w(C) \leq B$ and 
$\mS$ delivers the maximum queriability over $\mX$.
\end{problem}
\noindent
Unfortunately, the CECD problem cannot be solved in 
polynomial time in terms of input size unless $\mathbf{P} = \NP$.  

\begin{theorem}
\label{theorem:NPHard}
The problem of CECD is $\NP$-hard.
\end{theorem}
\begin{proof}
The problem of CECD can be reduced to the problem
of choosing cost-effective concepts from a set of concepts by
creating a taxonomy $\mX = (R, \mC, \mR)$ where all
nodes except for $R$ are leaf concepts, i.e. leaves.
Since the problem of choosing cost-effective concepts from 
a set of concepts is $\NP$-hard \cite{Termehchy:SIGMOD:14}, 
CECD will be $\NP$-hard. 
\end{proof}

\noindent
Because CECD is $\NP$-hard, we propose and study 
efficient approximation and 
pseudo-polynomial algorithms to solve it.

\section{Level-Wise Algorithm}
\label{sec:approximation-algorithms}
\prob{Level-wise} algorithm solves the problem of CECD using a greedy approach. 
It returns a design whose concepts are all from a same level 
of the input taxonomy.
Our algorithm finds
the design with maximum queriability for each level using the algorithm 
proposed in \cite{Termehchy:SIGMOD:14}, called
approximate popularity maximization ({\it APM} for short), 
for finding the cost-effective subset of concepts over 
a set of concepts. 
It eventually delivers the 
design with largest queriability across all levels
in the taxonomy.

Precisely, let $\mC[i]$ be the set of all concepts of 
depth $i$ in $\mX=$ $(R, \mC, \mR)$.  
For any concept $C\in \mC[i]$, 
we define its popularity $u(C)$  
to be the total popularity of its descendant 
leaf concepts in $\mX$.
Level-wise algorithm calls the APM algorithm 
to find the cost-effective
subset of concepts for every $\mC[i]$. 
It also computes the queriability of the design
that contains only the most popular leaf concept, i.e.,
the leaf concept with maximum $u$ value. 
It then compares various selected designs across $\mC[i]$s 
and returns the answer with maximum queriability 
as its solution for the problem of CECD over taxonomy $\mX$.
Figure~\ref{fig:level-or-max} illustrates the 
level-wise algorithm.  
Let $\card{\mC}$ denote the
number of concepts in taxonomy $\mX$. 
The APM algorithm runs in $O(\card{\mC} \log\card{\mC})$.
Thus, the time complexity of level-wise algorithm is
$O(h \card{\mC}\log\card{\mC})$ over taxonomy $\mX$.

In addition to being efficient, 
level-wise algorithm also has bounded and reasonably 
small worst-case approximation ratio for an
interesting case of CECD problem.
Sometimes, it may be easier 
to use and manage designs whose 
concepts are not subclass/ superclass of each other.
We call such a design a {\it disjoint design}.
Our empirical results in Section~\ref{sec:experiment}
shows that this strategy returns effective 
designs in the cases that the budget is relatively small.
In this case, we should restrict 
the feasible solutions in the CECD problem to be disjoint.
We call this case of CECD, {\it disjoint CECD}.

Recent empirical results suggest 
that the distribution of concept frequencies over a
large collection generally follows a \emph{power law} distribution 
\cite{Probase:Wu:12}. 
We show that the \prob{level-wise} algorithm 
has a bounded and reasonably small worst-case 
approximation ratio for CECD with disjoint design
given that distribution of concept 
frequencies follows a power law distribution.
The following lemma bounds the queriability
that is obtained from the free concepts in any
solution given that distribution of concept 
frequencies follows a power law distribution.
\begin{lemma}
\label{lemma:costly-intell}
Let $C_{\max}$ be the leaf concept in 
taxonomy $\mX=(R, \mC, \mR)$ with maximum $u$ value and let assume that distribution of $u$ over leaf concepts follows a \emph{power law} distribution. 
Let $\mS$ be any schema. Then,
\begin{align*}
QU(\free(\mS)) \leq 2u(\mC_{\max})\log\card{\mC}.
\end{align*}
\end{lemma}
\begin{proof}
We have:
$${\sum_{C\in \free(\mS)} u(C)d(C) \leq u(C_{\max})} \sum_{C\in \free(\mS)}d(C).$$
Since the frequencies of leaf concepts in $\mX$ follow a ``power law''
distribution,
\begin{align*}
\sum_{C \in \leaf(\mC)} d(C) \leq 1+\log(\card{\leaf(\mC)}),
\end{align*}
where $\leaf(\mC)$ is the set leaf concepts in $\mC$
and $\card{\leaf(\mC)}$ is the number of such concepts.
Since $|\leaf(\mC)|$ $\leq \card{\mC}$,
$$QU(\free(\mS)) \leq \sum_{C\in \free(\mS)} u(C)d(C) \leq (1 + \log\card{\mC})\  u(C_{\max}) \leq 2u(\mC_{\max})\log\card{\mC}.$$
\end{proof}

\begin{figure}[h]
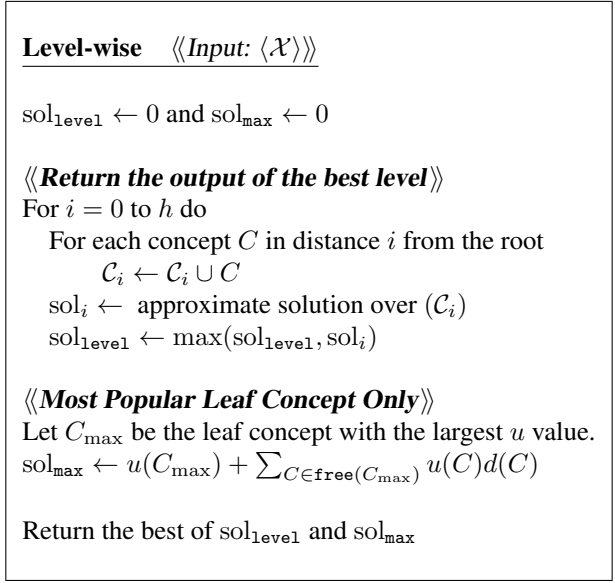

\begin{algo}
\\ \underline{\textbf{\prob{Level-wise}} \;\; \Comment{Input: $\left<\mX\right>$}}
\\
\\ $\sol_{\tt level} \leftarrow 0$ and $\sol_{\tt max} \leftarrow 0$
\\
\\ \Comment{\textbf{Return the output of the best level}}
\\ For $i=0$ to $h$ do
\\\> For each concept $C$ in distance $i$ from the root
\\\>\> $\mC_i \leftarrow \mC_i \cup C$
\\\> $\sol_{i} \leftarrow \text{ approximate solution over }(\mC_i)$
\\\> $\sol_{\tt level} \leftarrow \max(\sol_{\tt level}, \sol_{i})$
\\
\\ \Comment{\textbf{Most Popular Leaf Concept Only}}
\\ Let $C_{\max}$ be the leaf concept with the largest $u$ value.
\\ $\sol_{\tt max} \leftarrow u(C_{\max}) + \sum_{C\in\free(C_{\max})} u(C)d(C)$
\\ 
\\ Return the best of $\sol_{\tt level}$ and $\sol_{\tt max}$
\\
\end{algo}
\vspace{-0.5 cm}
\caption{Level-wise algorithm.}
\label{fig:level-or-max}
\end{figure}

\begin{theorem}
\label{theorem:approx}
Let $\mX = (R, \mC, \mR)$ be a taxonomy with height $h$
and the minimum accuracy of $\pr_{\min}$ $=\min_{C\in \mC}{\pr(C)}$.
The \prob{Level-wise} algorithm is a 
$O({h+\log\card{\mC}\over\pr_{\min}})$-approximation for the CECD problem with disjoint solution on $\mX$ and budget $B$
given that the distribution of frequencies in $\mC$ 
follows a power law distribution. 
\end{theorem}
\begin{proof}
Let $\mS^*$ be a disjoint schema over $\mX$ with total 
cost at most $B$ that maximizes $QU$ function.  
Let $\mS^*[i]$ be the set of concepts in $S^*$ of depth $i$.  
By the definition of disjointness, 
$\part(\mS^*[i])$ $\cap$ $\part(\mS^*[j]) = \emptyset$, 
for all $1\leq i,j\leq h$.  It follows:
\[
QU(\mS^*) = \sum_{1\leq i\leq h}{QU(\mS^*[i])} + QU(\free(\mS^*)),
\]
where $QU(\free(\mS^*)) = {\sum_{C\in \free(\mS^*)} u(C)d(C)}$
is the queriability obtained from the 
free concepts in $\mS^*$.

We consider two possible cases.  
First, assume that $\sum_{i=1}^h$ $QU(\mS^*[i])$ $\geq QU(\free(\mS^{*}))$.  
It immediately follows that the \prob{level-wise} algorithm output gives a $(2h/\pr_{\min})$-approximation. In the other case in which $QU(\free(\mS^{*})) \geq$ $\sum_{1\leq i\leq h} QU(\mS^*[i])$, by Lemma~\ref{lemma:costly-intell}, extracting the 
concept with the maximum $u$ value gives a $(4\log(\card{\mC})/$ $\pr_{\min})$-approximation.
These two cases together imply that we have an $O({h+\log\card{\mC}\over\pr_{\min}})$-approximation.
\end{proof}

\noindent
The value of $\pr_{\min}$ is generally large because concept annotation algorithm are reasonably accurate \cite{ChiticariuLRR10,McCALLUM:ACMQueue:05}.

\section{Pseudo-polynomial Time Algorithm}
\label{sec:pseudo-polynomial}
In this section we describe a pseudo-polynomial time 
algorithm for the CECD problem over tree taxonomies. 
As many other optimization problems on the tree structure, one approach is to find an 
optimal solution \emph{bottom-up} using dynamic programming technique.
The main idea is to define the CECD problem over all subtrees of the given taxonomy 
$\mX=(R, \mC,\mR)$. Next we show that in order to solve the subproblem 
defined over the subtree rooted at $C$, it is enough to solve the subproblems defined over 
the subtrees rooted at the children of $C$.

Let $\child(C)$ be the set of all children of the concept 
$C$ in $\mX$.
Moreover, let $\mX_C$ be the subtree of $\mX$ rooted 
at $C$. 
Formally given budget $B_C$, the subproblem over $\mX_C$ is to find a design $\mS_C \subseteq \mX_C$ whose total cost is at most $B_C$ and the queriability of the partitions obtained by $\mS_C$ is the maximum. Note that by annotating $\mS_C$ in $\mX_C$ there may exist a set of leaf concepts in $\mX_C$ that do not belong to any of $\part(S)$ for $S\in \mS$. Let $\nullpart(\mS_C, C)$ denotes the leaf concepts of $\mX_C$ that are not assigned to any partition of $\mS_C$. 

In order to computer the maximum queriability of the best design in $\mX_C$, one of the cases we should consider is the one in which $C$ is annotated.  To apply dynamic programming in this case we need to evaluate the queriability of $\part(C)$ which is $\sum_{Ch \in \child(C)}\sum_{C'\in \nullpart(\mS_{Ch}, Ch)} u(C')d(C')$. Thus besides the total queriability of partitions in $\mX_C$, we should compute the value of $\sum_{C' \in \nullpart(\mS_{Ch}, Ch)} u(C')d(C')$. All together we are required to solve the subproblem $Q$ defined over the subtree rooted at $C$ with parameter $B_C$ and $N_C$ where $B_C$ denotes the available budget for annotating concepts in $\mX_C$ and $N_C$ denotes the value of $\sum_{C'\in \nullpart(\mS_C, C)} u(C')d(C')$.

Further we assume that $u(C)$, $d(C)$, and $w(C)$ are positive integers for each $C \in \mC$. 
In Section~\ref{sec:experiment}, we show that the algorithm
can handle real values with scaling techniques in expense 
of reporting a near optimal solution instead of an optimal one.
We define $D = \sum_{C \in \leaf(\mC)} d(C)$, 
$U =$ $\sum_{C \in \leaf(\mC)} u(C)$. Let 
$B_{\tt total}$ denote the total available budget. 
We propose an algorithm whose time complexity is polynomial 
in $U$, $D$, $B_{\tt total}$, and $\card{\mC}$. 

We have the following 
recursive rules for the non-leaf concepts in $\mC$ based on the value of $Q$ for their children.

\begin{figure}[htb]
  \centering
    \includegraphics[height=1.2in,width=3.8in]{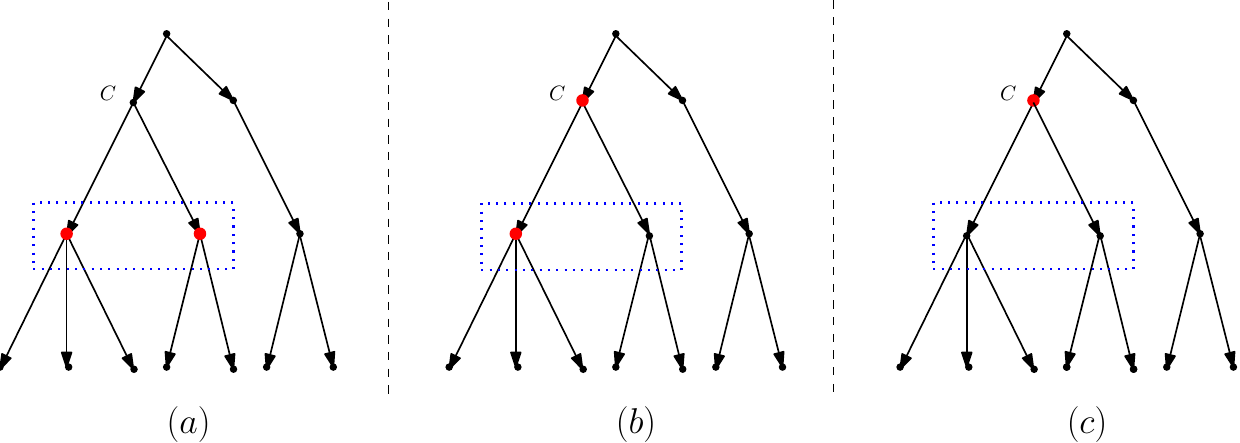}
      \caption{The concepts in red denote the ones that are picked in the design. $(a)$, $(b)$ and $(c)$ show three different types of the subproblems required to solve in order to compute $Q[C,B,0]$.} 
  \label{fig:dynamic-0}
\end{figure}

\begin{align*}
Q[C,B,0] = \max\{&\max_{\mB,\mN} (\sum_{Ch\in \child(C)} Q[Ch,\mB(Ch),\mN(Ch)] + \frac{\pr(C)}{d(C)} {\sum_{Ch \in \child(C)} \mN(Ch)}),\\ 
&\max_{\mB'}\sum_{Ch\in \child(C)} Q[Ch,\mB'(Ch),0]\}
\end{align*}
For each $Ch$, $\mB(Ch)$, $\mB'(Ch)$  and $\mN(Ch)$ are integer
values satisfying the following conditions:
(1) $B = w(C) + \sum_{Ch\in\child(C)}\mB(Ch)$,
(2) $B = \sum_{Ch\in\child(C)}\mB'(Ch)$, and
(3) $U D \geq \sum_{Ch\in\child(C)}\mN(Ch)$. 

The first term in the recursive rule corresponds to the case in which  
we select concept $C$ in the output design ($(b)$ and $(c)$ in Figure~\ref{fig:dynamic-0}) and 
the second term corresponds to the case in which 
for any child of $C$, $Ch$, $\nullpart(Ch)=\emptyset$ ($(a)$ in Figure~\ref{fig:dynamic-0}). 
In a design $\mS_C$ in $\mX_C$ with the maximum queriability and empty $\nullpart$ whose total cost is $B$, either $C$ is selected in the design and the budget $B-w(C)$ is divided among the children of $C$ (first term of the above rule), or the whole budget $B$ is divided among the children of $C$ and all leaf concepts of $\mX_C$ is assigned to a proper descendant of $C$ in the design (second term of the above rule).    

Similarly, for the case in which $N\neq 0$ we have the following recursive rule:
\begin{align*}
Q[C,B,N] &= \max_{\mB,\mN} \sum_{Ch\in \child(C)} Q[Ch,\mB(Ch),\mN(Ch)]
\end{align*}
where $\mB=\set{\mB(Ch)|Ch\in\child(C)}$ and $\mN=$ $\{\mN(Ch)|$ $Ch\in\child(C)\}$ such that 
$B = \sum_{Ch\in\child(C)}\mB(Ch)$ and 
$N =\sum_{Ch\in\child(C)}\mN(Ch)$. 
For each leaf concept $C_\ell$ in $\mC$, we have the following.
\squishlisttwo
\item{$Q[C_\ell,B,N] = 0 \text{ if } N=u(C_\ell)d(C_\ell) \text{ and } -\infty \text{ otherwise }$}

\item{$Q[C_\ell,B,0] = \pr(C_{\ell})u(C_\ell) \text{ if } B \geq w(C_\ell) \text{ and } -\infty \text{ otherwise}.$}
\end{list}
\noindent
The maximum value of the queriability on $\mX=(R, \mC, \mR)$ is
\begin{align}
%
\max_{N} Q[R,B_{\tt total},N] + N,
\label{eq:pseudo-final-recursion}
\end{align}
where $B_{\tt total}$ is the total available budget. The first term, $Q[R,B_{\tt total},N]$, denotes the profit obtained form the partitions of an optimal design and the second term corresponds to the profit obtained from the free concepts with respect to the output design.

To compute the running time of the algorithm we need to give an
upper bound on the number of cells in $Q$ and the time required 
to compute the value of each cell. 
The time to compute a single cell in $Q$ is
exponential in terms of the maximum degree of the taxonomy.  
Consequently, the algorithm runs much faster 
if the maximum degree in $\mX$ is bounded by a small constant.
As we show next, 
we can modify the taxonomy $\mX$ to obtain 
taxonomy $\mX'$ such that each concept $C$ in 
$\mX'$ has at most two children and the number of nodes in $\mX'$
is at most twice the number of nodes in $\mX$.
Since each node in $\mX'$ has two children, 
the required amount of time to compute a single cell in $Q$ is 
$O(B_{\tt total} U D)$; at most $B_{\tt total}$ ways to divide the budget between the two children and at most $UD$ ways to divide $N$ between the two children. 
Since the first argument in $Q$ can be any of the concepts 
in $\mC$, $N \leq U D$ and $B \leq B_{\tt total}$,
there are $O(B_{\tt total}UD)$ cells to evaluate in order 
to compute the design with maximum queribility.
Thus the total time for computing all cells
in $Q$ is $O(\card{\mC} (B_{\tt total} U D)^{2})$. 

Next, we explain how to transform an arbitrary taxonomy to 
a binary taxonomy. 
Let $C$ be a non-leaf concept in $\mX$.
We replace the induced subtree of $C\cup\child(C)$ with a 
full binary tree $\mX_C'$ whose root is $C$ and whose 
leaves are $\child(C)$ as shown in Figure~\ref{fig:binary-transform}.
Some internal nodes of $\mX_C'$ do not 
correspond to any node in $\mX$.  We refer to such 
internal nodes as {\it dummy nodes}, 
and set their cost to $B_{\tt total}+1$ to make sure that 
our algorithm does not include them in 
the output design.

\begin{figure}[htb]
  \centering
    \includegraphics[height=1.5in,width=3.75in]{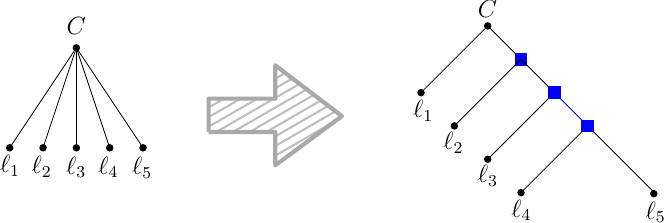}
      \caption{Transforming an input taxonomy $\mX$ into a binary taxonomy. 
		Blue square nodes correspond to dummy nodes.}
  \label{fig:binary-transform}
\end{figure}

Applying the mentioned transformation to all 
nodes of $\mX$, we obtain a binary taxonomy 
$\mX'= (R, \mC', \mR')$.  
The number of nodes in $\mC'$ is at most twice 
the number of nodes in $\mC$.
It follows that the running time of our pseudo-polynomial algorithm
on the input $\mX'$ is 
$O(\card{\mC} (B_{\tt total} U D)^{2})$. 

Since this transformation does not change the subset of 
leaf concepts in the subtree rooted in any internal node, 
any internal node in $\mX$ corresponds to a 
solution in $\mX'$ with the same cost and queriability.  
Since dummy nodes are too expensive to be chosen,
they do not introduce any new solution to the 
set of feasible solutions.
\begin{theorem}
There is an algorithm to solve the CECD problem over 
taxonomy $\mX = (R, \mC, \mR)$ with budget $B$ in 
$O(\card{\mC} B^2 U^2 D^2)$. 
\end{theorem}

Table~\ref{table:results} presents a summary of 
proposed algorithms for the CECD problem. 
\begin{table}[t]
\begin{center}
\begin{tabular}{|c|c|c|}
\hline
Algorithm & Approximation ratio & Running time \\
\hline
Level-wise &  
$O((h+\log\card{\mC})/\pr_{\min})$ (Disjoint CECD)& $O(h \card{\mC} \log(\card{\mC}))$\\
\hline
Dynamic Programming & Pseudo-polynomial & $O(\card{\mC} B^ 2 U^2 D^{2})$  \\
\hline
\end{tabular}
\end{center}
\vspace{-0.55cm}
\caption{Algorithms for the CECD problem.} 
\label{table:results}
\end{table}

\section{Cost-Effective Design for DAG Taxonomies}
\label{sec:dag-taxonomy}
\subsection{Directed Acyclic Graph Taxonomies}
While taxonomies are traditionally in form of trees,
many of them have evolved into \emph{directed acyclic graphs} (DAGs) 
to model more involved subclass/ superclass 
relationships between concepts
in their domains. Figure~\ref{fig:schema-org} shows fragments of
{\it schema.org} taxonomy. Some concepts in 
this taxonomy are included in multiple superclasses. 
For example,
a {\it hospital} is both a {\it place} and an {\it organization}. Therefore, a tree 
structure is not able to represent these relationships.

Formally, a {\it directed acyclic graph taxonomy} 
$\mX = $ $(R, \mC, \mR)$,  ({\it DAG taxonomy} for short), 
is a DAG, with vertex set $\mC$, edge set $\mR$, and \textit{root} $R$. 
$\mC$ is a set of concepts, 
$(D, C) \in \mR$ iff $D, C \in \mC$ and   
$D$ is a superclass of $C$.  Finally, $R$ is a node
in $\mX$ without any superclass.  
A concept $C \in \mC$ is a {\it leaf concept} iff it has no subclass in $\mX$; i.e, there is not 
any node $D \in \mC$ where $(C, D) \in \mR$.
The definitions of {\it child}, {\it ancestor},
and {\it descendant} 
over tree taxonomies naturally extends to DAG taxonomies.

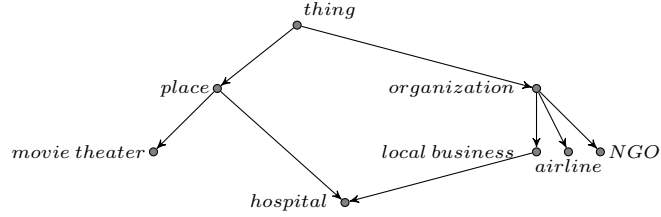
\begin{figure}[htb] 
\centering
\begin{tikzpicture}[->,>=stealth',scale=0.84]
\tikzstyle{every node}=[circle,draw,fill=black!50,inner sep=1pt,minimum width=3pt,font=\scriptsize]
\tikzstyle{every label}=[rectangle,draw=none,fill=none,font=\scriptsize]
\node (thing) [label=above right:$thing$] at (1.25,3) {};
\node (place) [label=left:$place$] at (0,2) {};
\node (org) [label=left:$organization\ \ \ $] at (5,2) {};
\node (airline) [label=below:$airline$] at (5.5,1) {};
\node (NGO) [label=right:$NGO$] at (6,1) {};
\node (local-business) [label= left:$local\ business\ \ \ $] at (5,1) {};
\node (movie-theater) [label=left:$movie\ theater$] at (-1,1) {};
\node (hospital) [label=left:$hospital\ \ $] at (2,0.2) {};
\draw [->] (thing) -- (place);
\draw [->] (place) -- (movie-theater);
\draw [->] (place) -- (hospital);
\draw [->] (thing) -- (org); 
\draw [->] (org) -- (local-business);
\draw [->] (local-business) -- (hospital);
\draw [->] (org) -- (airline);
\draw [->] (org) -- (NGO);  
\end{tikzpicture}
\caption{Fragments of {\it schema.org} taxonomy}
\label{fig:schema-org}
\end{figure}

\subsection{Design Queriability}
Design $\mS$ over DAG taxonomy $\mX = $ $(R, \mC, \mR)$ is 
a non-empty subset of $\mC - \set{R}$.
Due to the richer structure of DAG taxonomies, 
designs over DAG taxonomies may improve the effectiveness
of answering queries in more ways than the ones 
over tree taxonomies.
For example, let data set $DS$ be in the domain of 
the DAG taxonomy in Figure~\ref{fig:schema-org}, and $\mS_1= $ $\{place,$ $organization\}$ be a design. 
The query interface will examine 
the documents that are organized under {\it organization} 
in $DS$ to answer queries about concept {\it airline}. 
As query interface does not have sufficient information 
to pinpoint the entities of concept {\it airline} in $DS$,
it may return some non-relevant answers for these queries, 
e.g., matching entities that are NGOs.
On the other hand, because concept {\it hospital} 
is a subclass of both {\it place} and
{\it organization}, its entities in $DS$ are annotated 
by both concepts {\it place} and {\it organization}.
By examining the entities that are annotated by both
{\it place} and {\it organization}, the query interface
is able to identify the instances of {\it hospital} in $DS$.
Thus, it will not return entities that belong to other concepts
when answering queries about instances of {\it hospital}.
Generally, the query interface may pinpoint instances of some concepts in the data set by considering the 
intersections of multiple concepts in a design
over a DAG taxonomy. Hence, subsets of a design may 
create partitions in a DAG taxonomy. 
Next, we extend the notion of partitions for designs 
over DAG taxonomies.
\begin{definition}
\label{def:direct-ans}
Let $\mS$ be a design over DAG taxonomy $\mX = (R, \mC, \mR)$, and let $C \in \mC$ be a leaf concept.
An ancestor $A$ of $C$ in $\mS$ is $C$'s direct ancestor iff one of the following properties hold.
\squishlisttwo
\item $A = C$.
\item For each $D\in\mS$, if $D$ is an ancestor of $C$ then $D$ is not a descendant of $A$.
\end{list}
\end{definition}

The {\it full-ancestor-set} of $C$ is the set of {\it all its direct ancestors}.
For instance, the set $\{place,$ $organization\}$ is the full-ancestor-set of the 
concept {\it hospital} in design $\mS_1 =$ $\{place,$ $organization\}$, and the set 
$\{place,$ $local~business\}$ is the full-ancestor-set of the concept {\it hospital} in design
$\mS_2 =$ \{{\it place, organization, local business}\}
over the taxonomy in Figure~\ref{fig:schema-org}.

\begin{definition}
\label{def:partitionDAG}
Given design $\mS$ over DAG taxonomy $\mX =$ $(R, \mC, \mR)$, 
the partition of a set of concepts $\mathcal{D} \subseteq \mS$  is 
a set of leaf concepts $\mathcal{L} \subseteq \mC$ such that for every leaf concept
$L \in \mathcal{L}$, $\mathcal{D}$ is the full-ancestor-set of $L$. 

\end{definition}
\noindent
For instance, {\it hospital} belongs to 
the partition of $\{place,$\\ $organization\}$ 
in $\mS_1$. But, it does not belong to the partition of
$\set{place}$, since $\set{place}$ is not the full-ancestor-set of {\it hospital}.
The definitions of functions $\part$
and $\free$ over DAG taxonomies 
extend from their definitions over tree taxonomies.


Similar to tree taxonomies, we define the frequency of partition $P$, denoted by $d(P)$, 
as the frequency of the intersection of concepts in its root.
Using a similar analysis to the one in Section~\ref{sec:design-queriability}, we 
define the queriability of conceptual design $S$ over 
DAG taxonomy $\mX = $ $(R, \mC, \mR)$ as follows.
\begin{equation}
QU(\mS) = \sum_{P \in \allparts(\mS)}{\sum_{C\in P} u(C)d(C)  \over d(P)} + {\sum_{C\in \free(\mS)} u(C)d(C)}.
\label{eq:queriability-DAG}
\end{equation}
The function $\allparts(\mS) \subseteq 2^{\mS}$ returns the collection of all full-ancestor-sets of $\mS$ in $\mX$.  
We remark that the size of $\allparts(\mS)$ is linear, since we have at most one new partition per any leaf concept in $\mX$.
\noindent  

\subsection{Hardness of Cost-Effective Design Over DAG Taxonomies}
We define the CECD problem over DAG taxonomies similar to the 
CECD problem over tree taxonomies.
Following from the $\NP$-hardness results for CECD problem 
over tree taxonomy, 
CECD problem over DAG taxonomies is NP-hard. 
In this section, we prove that finding 
an approximation algorithm 
with a reasonably small bound on its approximation ratio 
for the problem CECD over DAG taxonomies 
is significantly hard. Unfortunately, this is true even for 
the special cases where concepts in the taxonomy 
have equal costs or the design is disjoint.  

We show that the CECD problem 
over a DAG taxonomy generalizes 
a hard problem in the approximation algorithms literature: \prob{Densest-$k$-Subgraph} \cite{Khot:FOCS:04}.
Given a graph $G=(V,E)$, in the 
the \prob{Densest-$k$-Subgraph} 
problem, the goal is to compute a subset $U\in V$ of size $k$ that 
maximizes the number of edges in the induced subgraph of $U$.
It is known that, unless $\mathbf{P}=\NP$, no polynomial time approximation 
scheme, i.e., PTAS, exists to compute the densest subgraph~\cite{Khot:FOCS:04}. 
Moreover, there are strong evidences that 
\prob{Densest-$k$-Subgraph} does not admit any approximation 
guarantee better than polylogarithmic factor \cite{Bhaskara:SODA:12,AroraDensest10}. 
The following theorem shows that approximating the $k$-densest subgraph reduces to approximating CECD.

\begin{lemma}\label{lem:annotate-improve}
Let $\mS$ be a design over taxonomy $\mX=(R,\mC,\mR)$ that is constructed from input $G=(V,E)$ as above. Let $S_v\in \mC\setminus \mS$ be a non-leaf concept. Then $QU(\mS\cup \set{S_v}) \geq QU(\mS)$. 
\end{lemma}
\begin{proof}
After annotating a non-leaf concept $S_v$, each leaf concept $C$ will be contained by a partition of either smaller or the same size. Since the contribution of a leaf concept $C$ to $QU$ only depends on the size of the partition contains $C$ and this dependence is a non-decreasing function in terms of the size of partition, after annotating $S_v$ the contribution of $C$ to $QU$ either increases or remains unchanged. Thus $QU(\mS \cup \set{S_v}) \geq QU(\mS)$.  
\end{proof}

\begin{figure}[htb]
\centering
\includegraphics[height=1.15in,width=3.5in]{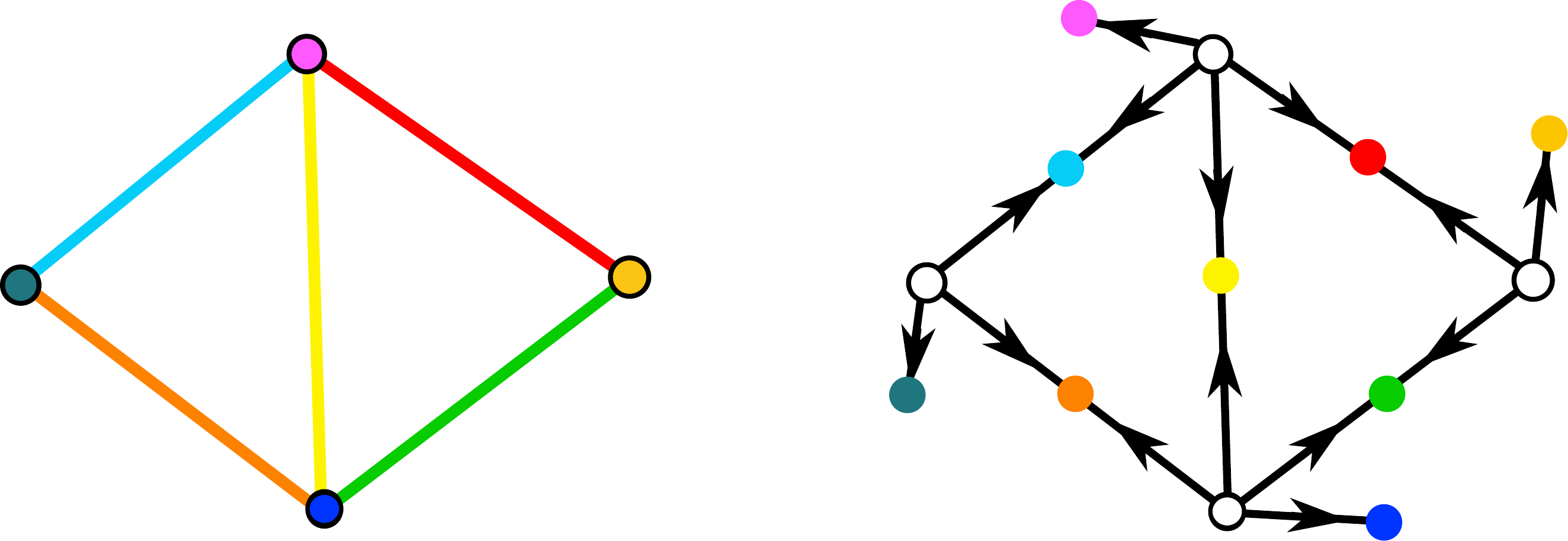}
\caption{Reducing the \prob{Densest-$k$-Subgraph} problem to CECD over DAG taxonomies where colors show correspondences in the reduction. The input graph for densest $k$-subgraph problem is shown on the left and its corresponding DAG taxonomies are in right. Colored vertices are leaf concepts and white vertices are non-leaf concepts in the DAG taxonomy.}
\label{fig:drawing}
\end{figure}

\noindent
The main result of this section is the following theorem. 
\begin{theorem}
\label{thm:log-apx}
A $(\log m)$-approximation algorithm for 
the CECD problem over DAG taxonomy with $m$ number of concepts implies that
there is an algorithm for the \prob{Densest-$k$-Subgraph} problem on $G=(V,E)$ 
with $n$ vertices that returns a $O(\log n)$-approximate solution. 
\end{theorem}
\begin{proof}
Given $G$ and $k$, we build an instance of the CECD over 
a DAG taxonomy as follows.
For each edge $e \in E$, we introduce a leaf concept $a_e$ 
and an for each vertex $v \in V$, 
we introduce a leaf concept $a_v$ and a non-leaf concept $S_v$ 
such that $S_v$ is the super class of $a_v$ and all the 
concepts corresponding to the incident edges to $v$ in $G$. 
Further, we set the budget $B$ to $k$, the cost of each non-leaf 
concept to $1$, and the cost of each leaf concept to $k+1$. 

Note that if we select $S_v$ and $S_u$ in the design and $(u,v) \in E$, then $a_{e}$ will be a singleton partition.
We also set the popularities and frequencies of all concepts in the taxonomy respectively to the same fixed values $u$ and $d$. Let $m$ be the number of edges in $G$ (or equivalently the number of leaf concepts in $\mC$) and $n$ be the number of vertices in $G$ (or equivalently the number of non-leaf concepts in $\mC$). For each partition $p\in \part(\mS)$ we set $d(p) = 1/(m\log n)$ if $\card{p} =1$ and $d(p)=1$ otherwise.

By Lemma~\ref{lem:annotate-improve}, annotating a non-leaf concept will not decrease the queribility of the design.
Since the leaf concepts are not affordable, and annotating a non-leaf concept will not decrease the total queribility, there exists an optimal design that annotates exactly $k$ non-leaf concepts. 
Note that in any design $\mS$ of size $k$, the contribution of any leaf concept in a non-singleton partition (partition of size greater than one) is exactly $u\cdot d$. 
In what follows we show that a $\log n$-approximation algorithm for the CECD problem implies a $O(\log n)$-approximation for the \prob{Densest-$k$-Subgraph} problem. To this end, by contradiction, let $\mA$ be a $\log n$-approximation algorithm of CECD problem.

Let $H_{\mS}$ be the set of vertices in $G$ of whose corresponding non-leaf concepts in $\mC$ are annotated in design $\mS$. $E(H_{\mS})$ denotes the set of edges with both endpoint in $H$ which corresponds to the set of edge-concepts of $\mC$ whose both non-leaf concepts corresponding to their endpoints are annotated by $\mS$. 
  
Let $\mS_{\opt}$ be an optimal solution of the CECD problem. Suppose that $QU(\mS_{\opt}) = (t  + r)\cdot m\log n + (m-t + n-r)$ where $t$ denotes the number of edges in $H_{\mS_{\opt}}$ and $r$ denotes the number of vertices in $H_{\mS_{\opt}}$ whose all incident edges are in $E(H_{\mS_{\opt}})$. It is straightforward to see that the corresponding leaf concepts to edges in $E(H_{\mS_{\opt}})$ and vertices with all incident edges in $E(H_{\mS_{\opt}})$ are the only singleton partitions with respect to design $\mS_{\opt}$. 

Now, let $\mS_{\mA}$ be the design returned by $\mA$ and similarly assume that $QU(\mS_{\mA}) = (t'+r')\cdot m\log n + (m-t' + n-r')$. Since $\mA$ is a $\log n$-approximation algorithm of the CECD problem, $(t  + r)\cdot m\log n + (m-t + n-r)$ is at most $\log n\cdot ( (t'+r')\cdot m\log n + (m-t' + n-r'))$. Thus,
\begin{align*}
t(m\log n -1) \leq t'(m\log^2n -1) + r'm\log^2 n + (m+n)\log n.
\end{align*}
Note that since the size of a feasible design is $k$, $r'\leq k$. Thus with some simplifications,
\begin{align*}
{tm\log n\over 2} \leq t'(m\log^2n) + km\log^2 n + 2m\log n,
\end{align*}
which implies that 
\begin{align}\label{eq:approx-ratio}
t \leq 2t'\log n + 2k\log n + 4 \leq 5\log n \cdot \max\set{k,t'}.
\end{align}

Now consider the greedy approach of \prob{Densest-$k$-Subgraph} problem such that in each step the algorithm picks a vertex $v$ and add it to the already selected set of vertices $S$ if $v$ has the maximum number of edges incident to $S$. It is easy to see that the greedy approach guarantee $k/2$ number of edges. Note that if the input graph has less than $k/2$ edges, we can solve \prob{Densest-$k$-Subgraph} problem optimally by picking all edges. 
Using the simple greedy approach and the result returned by $\mA$, we can find a set of $k$ vertices whose induced subgraph has at least $\max\set{k/2,t'}$ number of edges. Thus by (\ref{eq:approx-ratio}), we can find a $O(\log n)$-approximate solution of the \prob{Densest-$k$-Subgraph} problem which completes the proof. 
\end{proof}

Since the concepts in the instance of the 
CECD problem discussed in the proof of 
Theorem~\ref{thm:log-apx} have equal costs and
its optimal solution is disjoint, i.e., there is no 
directed path between any two of concepts in the design,   
the hardness results of Theorem~\ref{thm:log-apx}
is true even for the special cases of CECD problem over DAG taxonomies where the {\it concepts are equally costly} and/or 
{\it the problem has disjoint solutions}.

Figure~\ref{fig:DAG-instance} illustrates a simple example for which the \prob{level-wise} algorithm is arbitrarily 
worse than the optimal solution over DAG taxonomies. For the sake of 
simplicity, let $d$ and $u$ values be positive integers.
Let $u(C_4) = 4$, $d(C_4) = 1$, $u(C_5) = 1$, $d(C_5) = M$, $u(C_6) = M$, $d(C_6) = 1$, $u(C_7) = 1$ and $d(C_7) = M$.
Also, let $w(C_1) = w(C_2) = w(C_3) = 1$, and $B=2$.
The greedy algorithm first picks $C_1$ because of its high immediate queriability, and then $C_2$ or $C_3$ (but not both of them).
So its total queriability is $5$.  
On the other hand, by picking $C_2$ and $C_3$ one may 
acquire $C_6$ for free, whose queriability is $M$.
Since we can choose $M$ to be any number, 
the optimal solution can be arbitrarily better that the solution
delivered by the greedy approach.
Intuitively, the situation can be exacerbated to a 
large extent if the subset with large queriability
can be obtained by intersecting more than two concepts.

\begin{figure} 
\centering
\begin{tikzpicture}[->,>=stealth',scale=0.84]
\tikzstyle{every node}=[circle,draw,fill=black!50,inner sep=1pt,minimum width=3pt,font=\scriptsize]
\tikzstyle{every label}=[rectangle,draw=none,fill=none,font=\scriptsize]
\node (C0) [label=above:$C_0$] at (2,3) {};
\node (C1) [label=above:$C_1$] at (0,2) {};
\node (C2) [label=left:$C_2$] at (2,2) {};
\node (C3) [label=right:$C_3$] at (4,2) {};
\node (C4) [label=left:$C_4$; u:4 d:1] at (0,1) {};
\node (C5) [label=left:$C_5$; u:1 d:M] at (2,1) {};
\node (C6) [label=below:$C_6$; u:M d:1] at (3,1) {};
\node (C7) [label=right:$C_7$; u:1 d:M] at (4,1) {};
\draw [->] (C0) -- (C1);
\draw [->] (C0) -- (C2);
\draw [->] (C0) -- (C3);
\draw [->] (C1) -- (C4);
\draw [->] (C2) -- (C5);
\draw [->] (C2) -- (C6);
\draw [->] (C3) -- (C6); 
\draw [->] (C3) -- (C7);
\end{tikzpicture}
\caption{An instance of CECD problem over DAG taxonomy}
\label{fig:DAG-instance}
\end{figure}
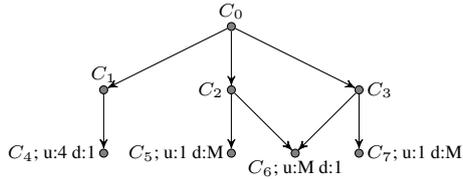

\newcommand{\mb}[1]{\textbf{#1}}
\newcommand{\mi}[1]{\textit{#1}}
\newcommand{\ml}[1]{\underline{#1}}
\newcommand{\mbi}[1]{\textit{\textbf{#1}}}
\newcommand{\mil}[1]{\underline{\textit{#1}}}
\newcommand{\mbl}[1]{\underline{\textbf{#1}}}
\newcommand{\mbil}[1]{\underline{\textit{\textbf{#1}}}}
\newcommand{\mo}[1]{\uwave{#1}}

\section{Experiments}
\label{sec:experiment}

\subsection{Experiment Setting}
\label{sec:expsetting}

\begin{table*}
\caption{The sizes and heights of taxonomies and the sizes 
of corresponding datasets and query workloads.
}
\label{tab:expinfo}
\centering
{\sffamily
\begin{tabular}{r|c|c|c|c|c|c|c|c}
Taxonomy & T1 & T2 & T3 & T4 & T5 & T6 & T7 & T8 \\
\hline
\#Concept	& 10 & 17 & 17 & 28 & 63 & 185 & 279 & 2269 \\
\#Height	& 2 & 2 & 3 & 7 & 6 & 8 & 8 & 9 \\
\#Documents	& 68982 & 267653 & 88479 & 1470661 & 1470661 & 1470661 & 1470661 & 1470661 \\
\#queries	& 648 & 256 & 146 & 4219 & 4888 & 4728 & 5216 & 11156 \\
\end{tabular}}
\end{table*}

\subsubsection{Taxonomies and datasets}

\noindent
{\bf Taxonomies:} We have extracted eight taxonomies from YAGO ontology 
version 2008-w40-2 \cite{Suchanek:YAGO} to validate 
our model and evaluate effectiveness and efficiency 
of our proposed algorithms.
YAGO organizes its concepts using subclass 
relationships in a DAG with a single root.
We have created the breath-first tree of YAGO 
and randomly selected the concepts from the tree for our taxonomies.
To validate our model, we have to  
compute and compare the effectiveness of 
answering queries using 
every feasible design over a taxonomy. 
Thus, we need tree taxonomies with 
relatively small number of concepts for our 
validation experiments.
We have extracted three tree taxonomies with relatively small numbers of 
nodes, called {\em T1}, {\em T2} and {\em T3}, 
to use in our validation experiments. 
The selected concepts for T1, T2 and T3 
are from level 3 to 6 of the full YAGO tree
which has a total of 9 levels.
We have further picked five other tree taxonomies with larger numbers of concepts, 
denoted as {\em T4}, {\em T5}, {\em T6}, {\em T7} and {\em T8}.
The concepts for T4 and T5 are selected from 
level 3 to 6 of the full YAGO tree. 
The concepts for T6 and T7 are randomly selected 
from levels 2 to 9 of the full YAGO tree.
T8 contains all concepts from the original YAGO tree that appear at least once
in the collection of English Wikipedia articles.
We use all tree taxonomies to evaluate 
the effectiveness and efficiency of our proposed algorithms.
Table~\ref{tab:expinfo} shows the properties of these taxonomies.

\noindent
{\bf Datasets:} 
We have used the collection of English Wikipedia articles 
from the Wikipedia dump\footnote{\it dumps.wikimedia.org} 
of October 8, 2008 
that is annotated by concepts from YAGO ontology 
in our experiments \cite{Suchanek:YAGO}.
This collection originally contains 2,666,190 
articles of which 1,470,661 articles are annotated 
by at least a concept from YAGO ontology.
For each taxonomy in our set of taxonomies, 
we have extracted a subset of the original 
Wikipedia collection where each document 
contains at least a mention to an entity 
of a concept in the taxonomy.
We use each dataset in the experiments 
over its corresponding taxonomy. 
Table~\ref{tab:expinfo} shows the properties of these 
eight datasets.
The annotation accuracies of the concepts 
in selected taxonomies over these datasets are 
between 0.8 and 0.95.

\subsubsection{Query Workload}
We use a subset of Bing\footnote{\it bing.com}, query log whose  
relevant answers are Wikipedia articles 
\cite{EvaluatingEvidences:Demidova}.
The relevant answers for each query in this query workload 
have been determined using the click-through information 
by eliminating noisy clicks. 
Each query contains between one to six keywords and has 
between one to two relevant answers 
with most queries having one relevant answer. 
Because the query log does not have the concepts behind 
its queries, we adapt an automatic approach to find 
the leaf concept from the taxonomy associated with each query.
We label each query by the concept of 
the matching instance in its relevant answer(s).
Then, we select queries labeled with a {\em single} concept.
\ignore{
If a query is labeled with $n>1$ concepts, we create $n$ copies of that query and label each of them with one of the $n$ concepts separately.
} 
Using this method, we create a query workload per 
each of our datasets. 
It is well known that the effectiveness of 
answering some queries may not be improved by 
annotating the dataset \cite{MoreSenses:Sanderson}.
For instance, all candidate answers for a query may 
contain mentions to the entities of the query concept.
To reasonably evaluate our algorithms, 
we have ignored the queries whose rankings 
remain the same over the unannotated version 
and the version of the dataset where 
all concepts in the taxonomy are annotated.
Table~\ref{tab:expinfo} shows the information about 
our query workloads. 
\ignore{
Table~\ref{tab:expinfo} shows the original number of 
queries with a single concept and 
the queries remaining after filtering out the
ones whose ranking quality is not improved by annotating
all concepts in the corresponding taxonomy.
}

We estimate the popularities of concepts for each taxonomy by sampling a small subset 
of randomly selected queries from their corresponding query workloads. 
We compute the popularity of each concept using estimation error rate of 
5\% under the 95\% confidence level. 
The number of sampled queries are between 100 and 500 for each taxonomy.
\ignore{
We use two-fold cross validation to calculate 
the popularities, $u$, of concepts in each taxonomy 
over their corresponding query workload.
} 
Because some concepts in a taxonomy may not appear 
in its query workload, we smooth the popularity of 
a concept $C$, $u(C)$, 
using the Bayesian $m$-estimate method \cite{IRBook}: 
$\hat{u}(C) = \frac{\hat{P}(C|QW)+mp}{m+\sum_C{\hat{P}(C|QW)}}$,
where $\hat{P}(C|QW)$ is the probability that 
$C$ occurs in the query workload $QW$ and 
$p$ denotes the prior probability. 
We set the value of the smoothing parameter, $m$, 
to 1 and use a uniform distribution for all the 
prior probabilities.

\subsubsection{Query Interface}
We index our datasets using Lucene\footnote{\it lucene.apache.org}. 
Given a query, we rank its candidate answers using 
BM25 ranking formula, which is shown to be more effective 
than other similar document ranking methods \cite{IRBook}.
Then, we apply the information about the concepts 
in the query and documents to return the answers
whose matching instances have the same concept as 
the concept of the query. 
If the concept in the query has not been annotated in 
the collection, the query interface returns the list of 
documents ranked by BM25 method without any modification. 
We have implemented our query interface and algorithms in 
Java 1.7 and performed our experiments on a Linux server 
with 100 GB of main memory and two quad core processors.

\subsubsection{Effectiveness Metric}

Because all queries in our query workloads have
one or two relevant answers, we measure the
ranking quality of answering queries over a dataset
using mean reciprocal rank (MRR) \cite{IRBook}.
MRR is an inverse of the rank of the first relevant answer 
in the returned list of answers.
MRR is used to measure the effectiveness of results for queries 
that have a small number of relevant answers \cite{IRBook}.
Since most queries in our query workload have a single relevant answer, 
MRR is an appropriate metric to measure the effectiveness of their results.

We also report the precision at top 3 answers ($p@3$) for our results.
Because all queries in our query workloads have
one or two relevant answers, we measure the
ranking quality of answering queries over a dataset
using precision at top 3 answers ($p@3$) \cite{IRBook}.
One may also use a larger number of top answers 
to measure the precision and ranking quality of the results. 
However, since the total numbers of relevant answers 
in our query workloads are less than three, 
the maximum possible value for $p@k$ is at most $\frac{2}{k}$. 
Consequently, the reported values will be very small when $k$ is large. 
Thus, using larger values of $k$ makes comparing and 
analyzing the effectiveness of query results rather hard. 
Further, using larger values of $k$ may hide the 
advantages of some query interfaces.
For instance, let a query have one relevant answer. 
Assume that the query interface $I1$ places the relevant 
answer to the query in the first position and the query 
interface $I2$ places it in the 10th position of 
their returned results. 
Intuitively, the result of $I1$ is more effective than $I2$. 
The values of $p@3$ are $0.33$ for $I1$ and $0$ 
for $I2$; however, 
the values of $p@10$ for both of them will be equal.

We measure the statistical significance of our results using 
the paired-$t$-test at a significant level of 0.05.

\subsubsection{Cost Models}

We use three models for generating costs of concept annotation.
First, we assign a randomly generated cost to each concept in a taxonomy.
The results reported for this model are averaged over 5 sets of random cost assignments per budget. 
We call this model {\it random cost} model. 
Second, when there is not any reliable estimation available for the cost of annotating concepts, 
an organization may assume that all concepts are equally costly. 
Hence, in our second cost model, we assume that 
all concepts in the input taxonomy have equal cost.
We name this model {\it uniform cost} model.
These two cost models have also been used to model the costs 
of large-scale data curation \cite{Dong:2012:LMS:2448936.2448938,Termehchy:SIGMOD:14,
Rekatsinas:2014:CSF:2588555.2610504}. 
Lastly, an organization may also assume that 
the cost of a concept depends on how likely
the concept appears in a collection.
Hence, in our third cost model, 
we assume that a cost of a concept 
is proportional to its frequency.
We call this model {\it frequency-based cost} model.
\ignore{
Lastly, sometimes, annotating a concept in a higher level of a taxonomy
can be easier than annotating a concept in a lower level.
For instance, annotating entities of concept {\em person} 
is easier than annotating entities of concept {\em criminal}.
Hence, in our fourth model, we set the cost of a concept 
to be proportional to its height in a taxonomy.
More specifically, we set a cost of a concept $c$ 
as $w(c) = w_1 + w_2 \times height(c)$.
In the formula, $height(c)$ denote a longest distance from a concept $c$ 
to any of its descendants that are leaves. 
$w_1$ and $w_2 \times height(c)$ represent a base cost 
and variable cost based on height when annotating $c$, respectively. 
We randomly generate both $w_1$ and $w_2$ per set of cost assignment.
The results reported for this model are averaged over 5 sets of cost assignments. 
We call this model {\it height-based cost} model.
}
We use a range of budgets between 0 and 1 with 
a step size of 0.1 where 1 means sufficient 
budget to annotate all leaf concepts in a taxonomy 
and 0 means no budget is available.

\subsection{Validating Queriability Function}
\label{sec:expvalidate}

\noindent {\bf Oracle:} 
Given a fixed budget, Oracle enumerates all 
feasible designs over the input taxonomy. Given an effectiveness metric, such as MRR, 
for each design, it computes the average the effectiveness metric for all queries in the query workload 
over the dataset annotated by the design.
It then returns the design with maximum value of the average effectiveness metric. 
\ignore{
Since oracle does not use any heuristic to predict the amount of improvement in effectiveness by a design, 
we use it to evaluate the accuracy of other methods that predict the degree of 
improvement in effectiveness achieved by a design.
Since each run of Oracle tests all possible designs over a taxonomy 
and runs all queries in the query workload per design, 
it may take several days to finish each run of Oracle over a taxonomy. 
Hence, we run Oracle only for the value of $p@3$.
}

\noindent {\bf Popularity Maximization ({\it PM}):} 
Following the traditional approach towards 
conceptual design for databases, 
one may select concepts in a design that are 
{\it more important} for users \cite{DBBook}.
Hence, we implement an algorithm, called {\em PM}, 
that enumerates all feasible designs, 
such as $\mS$, in a taxonomy and returns the one with 
the maximum value of 
$\sum_{p \in \part(\mS)} \sum_{C \in p}$ $u(C)\pr(p)$. 
This design contains the concepts that are 
more frequently queried by users and 
also annotated more accurately. 

\noindent {\bf Queriability Maximization ({\it QM}):} 
QM enumerates all feasible designs over the input taxonomy 
and returns the one with the maximum queriability as 
computed in Section~\ref{sec:design-queriability}.

Because we would like to explore how accurately PM and QM 
predict the amount of improvement in the effectiveness 
of answering queries by a design, 
we assume that these algorithms have complete information 
about the popularities and frequencies of concepts. 
Since Oracle, PM and QM algorithms enumerate 
all feasible designs, it is not possible to run them 
over large taxonomies. 
Hence, we run these algorithms over small taxonomies, 
i.e. T1, T2 and T3.  

Table~\ref{tab:model-others} illustrates the average MRR
achieved by Oracle, QM and PM over taxonomies T1, T2 and T3
for various budgets.
We do not report the values of average MRR for budgets 
greater than 0.7 for T1 and T2 and 0.6 for T3 
because all algorithms are equally effective.
The values of MRR show at budget 0.0 is the one achieved by 
BM25 ranking without annotating any concept in the datasets.
We note that there is no improvement in average MRR over
taxonomy T1 for budget 0.1 under uniform cost 
because each concept in the taxonomy costs more than 0.1.

\begin{table*}
\caption{Average $p@3$ for Oracle, PM and QM over T1, T2 and T3. 
Statistically significant differences between PM and QM and between Oracle and QM are marked in bold and italic, respectively. 
$B$ denotes a given budget.}
\label{tab:model-result}
\centering
{\scriptsize \sffamily
\begin{tabular}{r|c|ccc|ccc|ccc}
\multicolumn{1}{c}{} & & \multicolumn{3}{c}{Uniform Cost} & \multicolumn{3}{|c}{Random Cost} & \multicolumn{3}{|c}{Frequency-based Cost} \\
\cline{3-11}
\multicolumn{1}{c}{} & $B$ & Oracle & QM & PM & Oracle & QM & PM & Oracle & QM & PM \\
\hline\hline
\multirow{9}{*}{T1} 
&	0.0	&	\multicolumn{9}{|c}{0.089} \\ \cline{3-11}											
&	0.1	&	0.089	&	0.089	&	0.089	&	0.115	&	0.115	&	0.107	&	0.180	&	0.180	&	0.180	\\
&	0.2	&	0.149	&	\mb{0.149}	&	0.089	&	0.159	&	\mb{0.159}	&	0.102	&	0.180	&	0.180	&	0.180	\\
&	0.3	&	0.168	&	\mb{0.168}	&	0.091	&	0.170	&	\mb{0.170}	&	0.094	&	0.180	&	0.180	&	0.180	\\
&	0.4	&	0.183	&	\mb{0.177}	&	0.106	&	0.183	&	\mb{0.177}	&	0.116	&	0.180	&	0.180	&	0.180	\\
&	0.5	&	0.192	&	\mb{0.192}	&	0.166	&	0.192	&	\mb{0.192}	&	0.146	&	0.180	&	0.180	&	0.180	\\
&	0.6	&	0.194	&	\mb{0.193}	&	0.185	&	0.194	&	\mb{0.193}	&	0.177	&	0.180	&	0.180	&	0.180	\\
&	0.7	&	0.195	&	0.195	&	0.194	&	0.195	&	0.194	&	0.190	&	0.180	&	0.180	&	0.180	\\
\hline
\multirow{9}{*}{T2} 
&	0.0	&	\multicolumn{9}{|c}{0.200} \\ \cline{3-11}
&	0.1	&	0.241	&	0.232	&	0.234	&	\mi{0.252}	&	0.234	&	0.233	&	0.254	&	0.249	&	0.249	\\
&	0.2	&	0.275	&	\mb{0.272}	&	0.245	&	0.283	&	0.276	&	0.246	&	0.292	&	0.289	&	0.289	\\
&	0.3	&	\mi{0.303}	&	\mb{0.285}	&	0.247	&	\mi{0.309}	&	\mb{0.298}	&	0.248	&	0.303	&	0.293	&	0.293	\\
&	0.4	&	0.318	&	\mb{0.315}	&	0.250	&	0.318	&	\mb{0.316}	&	0.264	&	0.322	&	0.322	&	0.322	\\
&	0.5	&	0.320	&	\mb{0.318}	&	0.258	&	0.322	&	\mb{0.319}	&	0.278	&	0.323	&	0.322	&	0.322	\\
&	0.6	&	0.323	&	\mb{0.322}	&	0.290	&	0.323	&	0.321	&	0.299	&	0.323	&	0.323	&	0.323	\\
&	0.7	&	0.326	&	0.324	&	0.326	&	0.326	&	0.323	&	0.314	&	0.323	&	0.323	&	0.323	\\
\hline
\multirow{7}{*}{T3} 
&	0.0	&	\multicolumn{9}{|c}{0.171} \\ \cline{3-11}
&	0.1	&	0.222	&	0.210	&	0.208	&	0.240	&	0.235	&	0.235	&	0.210	&	0.210	&	0.210	\\
&	0.2	&	0.260	&	0.249	&	0.258	&	0.271	&	0.265	&	0.258	&	0.231	&	0.231	&	0.231	\\
&	0.3	&	0.281	&	0.269	&	0.258	&	0.292	&	\mb{0.289}	&	0.278	&	0.265	&	0.249	&	0.249	\\
&	0.4	&	0.304	&	0.304	&	0.288	&	0.304	&	\mb{0.304}	&	0.293	&	0.288	&	0.288	&	0.288	\\
&	0.5	&	0.306	&	0.304	&	0.299	&	0.306	&	0.304	&	0.303	&	0.304	&	0.304	&	0.304	\\
&	0.6	&	0.306	&	0.306	&	0.306	&	0.306	&	0.306	&	0.306	&	0.304	&	0.304	&	0.304	\\
\end{tabular}}
\end{table*}
\begin{table*}
\caption{Average MRR for Oracle, PM and QM over T1, T2 and T3. 
Statistically significant differences between PM and QM and between Oracle and QM are marked in bold and italic, respectively.
$B$ denotes a given budget.}
\label{tab:model-others}
\centering
{\scriptsize \sffamily
\begin{tabular}{r|c|ccc|ccc|ccc}
\multicolumn{1}{c}{} & & \multicolumn{3}{c}{Uniform Cost} & \multicolumn{3}{|c}{Random Cost} & \multicolumn{3}{|c}{Frequency-based Cost} \\
\cline{3-11}
\multicolumn{1}{c}{} & $B$ & Oracle & QM & PM & Oracle & QM & PM & Oracle & QM & PM \\
\hline\hline
\multirow{8}{*}{T1} 
&	0.0	&	\multicolumn{9}{|c}{0.187} \\ \cline{3-11}											
&	0.1	&	0.187	&	0.187	&	0.187	&	0.259	&	0.255	&	0.239	&	0.475	&	0.475	&	0.475	\\
&	0.2	&	0.362	&	\mb{0.362}	&	0.197	&	0.385	&	\mb{0.384}	&	0.226	&	0.475	&	0.475	&	0.475	\\
&	0.3	&	0.415	&	\mb{0.406}	&	0.203	&	0.424	&	\mb{0.417}	&	0.212	&	0.475	&	0.475	&	0.475	\\
&	0.4	&	0.459	&	\mb{0.459}	&	0.227	&	0.461	&	\mb{0.461}	&	0.262	&	0.475	&	0.475	&	0.475	\\
&	0.5	&	0.492	&	\mb{0.492}	&	0.400	&	0.492	&	\mb{0.492}	&	0.341	&	0.475	&	0.475	&	0.475	\\
&	0.6	&	0.501	&	\mb{0.501}	&	0.444	&	0.501	&	\mb{0.499}	&	0.426	&	0.475	&	0.475	&	0.475	\\
&	0.7	&	0.507	&	0.507	&	0.497	&	0.507	&	\mb{0.504}	&	0.476	&	0.475	&	0.475	&	0.475	\\
\hline
\multirow{8}{*}{T2} 
&	0.0	&	\multicolumn{9}{|c}{0.342} \\ \cline{3-11}											
&	0.1	&	0.504	&	0.479	&	0.504	&	0.515	&	0.471	&	0.482	&	0.598	&	0.598	&	0.598	\\
&	0.2	&	0.577	&	0.543	&	0.551	&	0.616	&	0.586	&	0.559	&	0.662	&	0.662	&	0.662	\\
&	0.3	&	\mi{0.641}	&	\mb{0.629}	&	0.574	&	0.677	&	\mb{0.661}	&	0.576	&	0.677	&	0.674	&	0.674	\\
&	0.4	&	0.729	&	\mb{0.729}	&	0.586	&	0.725	&	\mb{0.725}	&	0.616	&	0.741	&	0.741	&	0.741	\\
&	0.5	&	0.745	&	\mb{0.745}	&	0.615	&	0.744	&	\mb{0.742}	&	0.652	&	0.747	&	0.747	&	0.747	\\
&	0.6	&	0.751	&	\mb{0.751}	&	0.647	&	0.754	&	0.754	&	0.695	&	0.748	&	0.748	&	0.748	\\
&	0.7	&	0.763	&	0.763	&	0.759	&	0.763	&	0.758	&	0.732	&	0.748	&	0.748	&	0.748	\\
\hline
\multirow{7}{*}{T3} 
&	0.0	&	\multicolumn{9}{|c}{0.322} \\ \cline{3-11}											
&	0.1	&	0.469	&	0.469	&	0.453	&	0.528	&	0.521	&	0.514	&	0.447	&	0.447	&	0.447	\\
&	0.2	&	0.595	&	0.594	&	0.579	&	0.646	&	\mb{0.629}	&	0.587	&	0.531	&	0.531	&	0.531	\\
&	0.3	&	0.680	&	\mb{0.679}	&	0.600	&	0.714	&	\mb{0.712}	&	0.660	&	0.594	&	0.594	&	0.594	\\
&	0.4	&	0.745	&	\mb{0.734}	&	0.685	&	0.747	&	0.739	&	0.711	&	0.725	&	0.725	&	0.725	\\
&	0.5	&	0.754	&	0.754	&	0.741	&	0.758	&	0.757	&	0.748	&	0.734	&	0.734	&	0.734	\\
&	0.6	&	0.760	&	0.760	&	0.760	&	0.760	&	0.760	&	0.760	&	0.734	&	0.734	&	0.734	\\
\end{tabular}}
\end{table*}

Over all taxonomies and cost models, the designs selected by QM 
deliver close MRR values to the ones selected by Oracle.
There are a few cases where the results of QM are significantly 
worse than the results of Oracle.
For instance, consider the results of QM and Oracle for 
budget 0.3 over taxonomy T2.
In T2, concept {\em writing} 
is the parent of a leaf concept {\em dramatic composition}
and a couple other leaf concepts whose popularities are much less than that of 
{\em dramatic composition}.
QM picks a design that contains {\em dramatic composition}. 
This design will deliver the highest values of MRR
for queries with concept {\em dramatic composition}, 
but it does {\it not} help improving the values of MRR
for queries whose concepts are other children of {\em writing}
over the unannotated dataset.
That is, QM picks a relatively less popular concept, but maximizes 
the improvement of the effectiveness for queries with this concept.
On the other hand, Oracle selects {\em writing}
instead of {\em dramatic composition}.
Intuitively, this design improves the values of MRR  
for queries with concept {\em dramatic composition} 
less than the design selected by QM.
However, this design will improve the values of MRR  
for queries with other child concepts of {\em writing}.
For this dataset, selecting {\em writing} helps 
improving the values of MRR for queries 
with concept {\em dramatic composition} 
as equal as selecting {\em dramatic composition}.
Because, the design selected by QM is not able to 
improve the effectiveness of answering queries whose concepts are 
other children of {\em writing},
QM is less effective than Oracle.
We have observed a similar behavior for other cases 
when the results of QM are significantly worse than Oracle.
This observation suggests that if the budget is relatively small, 
it is sometimes better to annotate rather general concepts.

We must note that, for frequency-based cost, 
Oracle, QM and PM are equally effective.
This is because, under this cost model, 
the leaf concepts are cheaper than any internal concept node.
Hence, every algorithm chooses leaf concepts 
and returns the same design.
Furthermore, Oracle, QM and PM return the same values of average MRR 
over T1 for all budgets. 
This is because these algorithms can select 8 out of 9 leaf concepts 
using budget of 0.1, and the returned designs answer queries 
with those concepts effectively.
However, the remaining leaf concept, e.g., {\em person}, 
costs more than 0.9 because of its 92\% frequency in T1.
Because each internal concepts either costs more than {\em person}
or has all of their leaf descendants included in a design, 
these algorithms choose all leaf concepts except {\em person}
and do not include any internal concept.
Therefore, over T1, the algorithms are equally effective 
across all test budgets under frequency-based cost model.

Nevertheless, 
the results from Table~\ref{tab:model-others} indicate that 
QM delivers designs that improve the average MRR 
of answering queries more than the ones picked by PM. 
Overall, PM annotates more general concepts from the taxonomy
to improve the effectiveness of larger number of queries. 
Hence, to answer a query, the query interface often has to 
examine the documents annotated by an ancestor of the query concept.
As this set of documents contain many answers 
whose concepts are different from the query concept, 
the query interface is usually not able to improve
the value of MRR for a query significantly. 
QM selects the designs with relatively less general concepts.
Although its designs may not improve the ranking quality of 
every query, the designs significantly improve the ranking 
quality of queries whose concepts belong to the selected designs.

Table~\ref{tab:model-result} shows the average $p@3$ 
delivered by the designs returned by Oracle, QM and PM
over taxonomies T1, T2 and T3 
under uniform and random cost models for various budgets.
Overall, 
QM delivers designs with $p@3$ values close to the ones selected by Oracle.
Also, the average $p@3$ values for the designed selected by QM 
are generally higher than those delivered by PM.

\subsection{Effectiveness of the Proposed Algorithm}
\label{sec:approxalg}

\begin{table*}
\caption{Average $p@3$ for APM, APM-L, LW and DP$_{\epsilon=0.05}$
over T1, T2, T3, T4 and T5.
Statistically significant difference between APM and LW, 
between APM and DP, between DP and LW, between APM-L and LW, and between DP and APM-L 
are in italic, bold, underline, apostrophe(') and star(*), respectively.
$B$ denotes a given budget.}
\label{tab:approx-prec3}
\centering
{\scriptsize \sffamily
\begin{tabular}{r|c*{3}{|l@{\hspace{0.7em}}l@{\hspace{0.7em}}l@{\hspace{0.7em}}l}}
\multicolumn{1}{c}{} & & \multicolumn{4}{c}{Uniform Cost} & \multicolumn{4}{|c}{Random Cost} & \multicolumn{4}{|c}{Frequency-based Cost} \\
\cline{3-14}
\multicolumn{1}{@{\hspace{0em}}c}{} & $B$ & APM & APM-L & LW & DP & APM & APM-L & LW & DP & APM & APM-L & LW & DP \\
\hline\hline
\multirow{9}{*}{T1}														
&	0.1	&	0.089	&	0.088	&	0.089	&	0.089	&	0.107	&	0.113	&	0.113	&	0.115	&	0.180	&	0.180	&	0.180	&	0.180	\\
&	0.2	&	0.089	&	0.103	&	\mi{0.103}	&	\mbl{0.149}*	&	0.092	&	0.137	&	\mi{0.137}	&	\mbl{0.159}*	&	0.180	&	0.180	&	0.180	&	0.180	\\
&	0.3	&	0.103	&	0.164	&	\mi{0.164}	&	\mb{0.168}	&	0.125	&	0.161	&	0.161	&	\mbl{0.170}*	&	0.180	&	0.180	&	0.180	&	0.180	\\
&	0.4	&	0.164	&	0.183	&	\mi{0.183}	&	\mb{0.177}	&	0.149	&	0.183	&	\mi{0.183}	&	\mb{0.177}	&	0.180	&	0.180	&	0.180	&	0.180	\\
&	0.5	&	0.183	&	0.192	&	\mi{0.192}	&	\mb{0.192}	&	0.175	&	0.192	&	\mi{0.192}	&	\mb{0.189}	&	0.180	&	0.180	&	0.180	&	0.180	\\
&	0.6	&	0.192	&	0.193	&	0.193	&	0.193	&	0.188	&	0.192	&	0.192	&	0.189	&	0.180	&	0.180	&	0.180	&	0.180	\\
&	0.7	&	0.193	&	0.193	&	0.193	&	0.193	&	0.192	&	0.193	&	0.193	&	0.193	&	0.180	&	0.180	&	0.180	&	0.180	\\
&	0.8	&	0.193	&	0.195	&	0.195	&	0.193	&	0.193	&	0.193	&	0.193	&	0.193	&	0.180	&	0.180	&	0.180	&	0.180	\\
&	0.9	&	0.195	&	0.195	&	0.195	&	0.195	&	0.193	&	0.195	&	0.195	&	0.195	&	0.180	&	0.180	&	0.180	&	0.180	\\
\hline														
\multirow{9}{*}{T2}														
&	0.1	&	0.234	&	0.232	&	0.232	&	0.232	&	0.231	&	0.233	&	0.233	&	0.229	&	0.247	&	0.249	&	0.247	&	0.249	\\
&	0.2	&	0.242	&	0.245	&	0.245	&	0.245	&	0.255	&	0.262	&	0.262	&	0.271	&	0.251	&	0.289	&	\mi{0.289}	&	\mb{0.289}	\\
&	0.3	&	0.253	&	0.285	&	\mi{0.285}	&	\mb{0.285}	&	0.264	&	0.290	&	\mi{0.290}	&	\mb{0.292}	&	0.293	&	0.293	&	0.293	&	0.292	\\
&	0.4	&	0.292	&	0.316	&	\mi{0.316}	&	\mb{0.315}	&	0.295	&	0.301	&	0.301	&	\mb{0.310}	&	0.293	&	0.322	&	\mi{0.322}	&	\mb{0.320}	\\
&	0.5	&	0.294	&	0.318	&	\mi{0.318}	&	\mb{0.318}	&	0.298	&	0.320	&	\mi{0.320}	&	\mb{0.313}	&	0.309	&	0.322	&	\mi{0.322}	&	\mb{0.322}	\\
&	0.6	&	0.323	&	0.320	&	0.322	&	0.320	&	0.306	&	0.321	&	0.321	&	0.320	&	0.309	&	0.323	&	\mi{0.323}	&	\mb{0.323}	\\
&	0.7	&	0.323	&	0.326	&	0.326	&	0.324	&	0.323	&	0.323	&	0.323	&	0.323	&	0.323	&	0.323	&	0.323	&	0.323	\\
&	0.8	&	0.323	&	0.326	&	0.326	&	0.326	&	0.324	&	0.324	&	0.325	&	0.324	&	0.323	&	0.323	&	0.323	&	0.323	\\
&	0.9	&	0.323	&	0.326	&	0.326	&	0.326	&	0.325	&	0.325	&	0.326	&	0.324	&	0.323	&	0.323	&	0.323	&	0.323	\\
\hline														
\multirow{9}{*}{T3}														
&	0.1	&	0.208	&	0.210	&	0.210	&	0.210	&	0.215	&	0.216	&	0.224'	&	0.232*	&	0.192	&	0.192	&	0.192	&	\mbl{0.210}*	\\
&	0.2	&	0.208	&	0.249	&	\mi{0.249}	&	\mb{0.249}	&	0.242	&	0.253	&	0.262'	&	0.265*	&	0.231	&	0.231	&	0.231	&	0.231	\\
&	0.3	&	0.258	&	0.269	&	0.269	&	0.269	&	0.265	&	0.280	&	\mi{0.287}'	&	\mb{0.283}	&	0.249	&	0.249	&	0.249	&	0.249	\\
&	0.4	&	0.288	&	0.297	&	\mi{0.304}	&	\mb{0.304}	&	0.280	&	0.297	&	\mi{0.297}	&	\mb{0.304}	&	0.265	&	0.269	&	0.269	&	\mbl{0.285}*	\\
&	0.5	&	0.297	&	0.304	&	\mi{0.306}	&	\mb{0.306}	&	0.291	&	0.305	&	\mi{0.305}	&	\mb{0.305}	&	0.290	&	0.304	&	\mi{0.304}	&	\mb{0.304}	\\
&	0.6	&	0.304	&	0.306	&	0.306	&	0.306	&	0.302	&	0.305	&	0.306	&	0.306	&	0.297	&	0.304	&	0.304	&	0.304	\\
&	0.7	&	0.304	&	0.306	&	0.306	&	0.306	&	0.302	&	0.306	&	0.306	&	0.306	&	0.279	&	0.306	&	\mi{0.306}	&	\mb{0.306}	\\
&	0.8	&	0.306	&	0.306	&	0.306	&	0.306	&	0.305	&	0.306	&	0.306	&	0.306	&	0.290	&	0.306	&	\mi{0.306}	&	\mb{0.306}	\\
&	0.9	&	0.306	&	0.306	&	0.306	&	0.306	&	0.306	&	0.306	&	0.306	&	0.306	&	0.306	&	0.306	&	0.306	&	0.306	\\
\hline														
\multirow{9}{*}{T4}														
&	0.1	&	0.158	&	0.172	&	\mi{0.172}	&	\mb{0.172}	&	0.173	&	0.175	&	0.180'	&	0.176	&	0.157	&	0.157	&	0.157	&	0.161	\\
&	0.2	&	0.160	&	0.187	&	\mi{0.187}	&	\mb{0.187}	&	0.175	&	0.189	&	\mi{0.189}	&	\mb{0.188}	&	0.164	&	0.164	&	0.164	&	\mbl{0.171}*	\\
&	0.3	&	0.177	&	0.197	&	\mi{0.197}	&	\mb{0.200}	&	0.179	&	0.203	&	\mi{0.206}	&	\mb{0.203}	&	0.171	&	0.172	&	0.171	&	\mbl{0.196}*	\\
&	0.4	&	0.178	&	0.219	&	\mi{0.219}	&	\mb{0.215}	&	0.180	&	0.214	&	\mi{0.214}	&	\mb{0.215}	&	0.208	&	0.217	&	0.217	&	0.213	\\
&	0.5	&	0.185	&	0.223	&	\mi{0.223}	&	\mbl{0.232}*	&	0.189	&	0.227	&	\mi{0.227}	&	\mb{0.228}	&	0.215	&	0.225	&	0.225	&	0.225	\\
&	0.6	&	0.197	&	0.235	&	\mi{0.235}	&	\mb{0.235}	&	0.203	&	0.234	&	\mi{0.234}	&	\mb{0.231}	&	0.216	&	0.233	&	\mi{0.233}	&	\mb{0.233}	\\
&	0.7	&	0.215	&	0.240	&	\mi{0.240}	&	\mb{0.239}	&	0.213	&	0.239	&	\mi{0.239}	&	\mb{0.238}	&	0.208	&	0.233	&	\mi{0.233}	&	\mb{0.233}	\\
&	0.8	&	0.216	&	0.240	&	\mi{0.240}	&	\mb{0.239}	&	0.218	&	0.240	&	\mi{0.240}	&	\mb{0.238}	&	0.211	&	0.237	&	\mi{0.237}	&	\mb{0.235}	\\
&	0.9	&	0.218	&	0.241	&	\mi{0.241}	&	\mb{0.239}	&	0.220	&	0.241	&	\mi{0.241}	&	\mb{0.240}	&	0.219	&	0.240	&	\mi{0.240}	&	\mb{0.237}	\\
\hline														
\multirow{9}{*}{T5} 
&	0.1	&	0.173	&	0.183	&	\mi{0.183}	&	\mb{0.188}	&	0.183	&	0.192	&	0.194'	&	0.191	&	0.168	&	0.167	&	0.170'	&	0.169	\\
&	0.2	&	0.190	&	0.207	&	\mi{0.207}	&	\mb{0.202}	&	0.193	&	0.210	&	\mi{0.213}	&	\mb{0.210}	&	0.174	&	0.174	&	\mi{0.181}'	&	\mb{0.182}*	\\
&	0.3	&	0.206	&	0.222	&	\mi{0.222}	&	\mb{0.222}	&	0.212	&	0.224	&	\mil{0.226}	&	0.219	&	0.182	&	0.214*	&	\mi{0.214}	&	\mb{0.197}	\\
&	0.4	&	0.217	&	0.234*	&	\mil{0.234}	&	\mb{0.227}	&	0.222	&	0.232	&	\mi{0.235}	&	\mb{0.232}	&	0.219	&	0.238	&	\mi{0.238}	&	\mb{0.230}	\\
&	0.5	&	0.224	&	0.240	&	\mi{0.240}	&	\mb{0.235}	&	0.228	&	0.240	&	\mi{0.241}	&	\mb{0.238}	&	0.226	&	0.243	&	\mi{0.243}	&	\mb{0.243}	\\
&	0.6	&	0.235	&	0.245	&	\mi{0.245}	&	0.241	&	0.236	&	0.245	&	\mi{0.245}	&	\mb{0.241}	&	0.233	&	0.244	&	\mi{0.243}	&	\mb{0.244}	\\
&	0.7	&	0.236	&	0.247	&	\mi{0.247}	&	\mb{0.247}	&	0.239	&	0.248	&	0.248	&	0.245	&	0.221	&	0.247	&	\mi{0.247}	&	\mb{0.247}	\\
&	0.8	&	0.241	&	0.249	&	0.249	&	0.249	&	0.242	&	0.249	&	0.249	&	0.244	&	0.227	&	0.249	&	\mi{0.249}	&	\mb{0.248}	\\
&	0.9	&	0.245	&	0.250	&	0.250	&	0.250	&	0.246	&	0.250	&	0.250	&	0.250	&	0.223	&	0.250	&	\mi{0.250}	&	\mb{0.250}	\\
\end{tabular}}
\end{table*}
\begin{table*}
\caption{Average MRR for APM, APM-L, LW and DP$_{\epsilon=0.05}$
over T1, T2, T3, T4 and T5.
Statistically significant difference between APM and LW, 
between APM and DP, between DP and LW, between APM-L and LW, and between DP and APM-L 
are in italic, bold, underline, apostrophe(') and star(*), respectively.
$B$ denotes a given budget.}
\label{tab:approx-mrr}
\centering
{\scriptsize \sffamily
\begin{tabular}{r|c*{3}{|l@{\hspace{0.7em}}l@{\hspace{0.7em}}l@{\hspace{0.7em}}l}}
\multicolumn{1}{c}{} & & \multicolumn{4}{c}{Uniform Cost} & \multicolumn{4}{|c}{Random Cost} & \multicolumn{4}{|c}{Frequency-based Cost} \\
\cline{3-14}
\multicolumn{1}{@{\hspace{0em}}c}{} & $B$ & APM & APM-L & LW & DP & APM & APM-L & LW & DP & APM & APM-L & LW & DP \\
\hline\hline
\multirow{9}{*}{T1} 
&	0.1	&	0.187	&	0.187	&	0.187	&	0.187	&	0.239	&	0.250	&	0.250	&	0.255	&	0.475	&	0.475	&	0.475	&	0.475	\\
&	0.2	&	0.197	&	0.220	&	\mi{0.220}	&	\mbl{0.362}*	&	0.204	&	0.318	&	\mi{0.318}	&	\mb{0.384}	&	0.475	&	0.475	&	0.475	&	0.475	\\
&	0.3	&	0.221	&	0.394	&	\mi{0.394}	&	\mb{0.406}	&	0.288	&	0.390	&	\mi{0.390}	&	\mbl{0.417}*	&	0.475	&	0.475	&	0.475	&	0.475	\\
&	0.4	&	0.394	&	0.438	&	\mi{0.438}	&	\mbl{0.459}*	&	0.357	&	0.440	&	\mi{0.440}	&	\mbl{0.460}*	&	0.475	&	0.475	&	0.475	&	0.475	\\
&	0.5	&	0.438	&	0.492	&	\mi{0.492}	&	\mb{0.492}	&	0.421	&	0.492	&	\mi{0.492}	&	\mb{0.486}	&	0.475	&	0.475	&	0.475	&	0.475	\\
&	0.6	&	0.492	&	0.501	&	0.501	&	0.492	&	0.471	&	0.496	&	\mi{0.496}	&	0.488	&	0.475	&	0.475	&	0.475	&	0.475	\\
&	0.7	&	0.501	&	0.502	&	0.502	&	0.493	&	0.494	&	0.502	&	\mi{0.502}	&	0.494	&	0.475	&	0.475	&	0.475	&	0.475	\\
&	0.8	&	0.502	&	0.507	&	0.507	&	0.507	&	0.502	&	0.503	&	0.503	&	0.502	&	0.475	&	0.475	&	0.475	&	0.475	\\
&	0.9	&	0.507	&	0.507	&	0.507	&	0.507	&	0.502	&	0.506	&	0.506	&	0.506	&	0.475	&	0.475	&	0.475	&	0.475	\\
\hline														
\multirow{9}{*}{T2} 
&	0.1	&	0.504	&	0.479	&	0.479	&	0.479	&	0.481	&	0.466	&	0.466	&	0.476	&	0.589	&	0.598	&	0.589	&	0.598	\\
&	0.2	&	0.534	&	0.566	&	\mi{0.566}	&	\mb{0.566}	&	0.564	&	0.567	&	0.571'	&	\mbl{0.591}*	&	0.595	&	0.662	&	\mi{0.662}	&	\mb{0.662}	\\
&	0.3	&	0.580	&	0.629	&	\mi{0.629}	&	\mb{0.629}	&	0.607	&	0.638	&	\mi{0.638}	&	\mbl{0.652}*	&	0.667	&	0.668	&	0.668	&	0.671	\\
&	0.4	&	0.651	&	0.718	&	\mi{0.718}	&	\mb{0.729}	&	0.666	&	0.686	&	\mi{0.686}	&	\mbl{0.713}*	&	0.668	&	0.741	&	\mi{0.741}	&	\mb{0.739}	\\
&	0.5	&	0.673	&	0.745	&	\mi{0.745}	&	\mb{0.734}	&	0.673	&	0.738*	&	\mil{0.738}	&	0.728	&	0.703	&	0.747	&	\mi{0.747}	&	\mb{0.747}	\\
&	0.6	&	0.746	&	0.751	&	0.751	&	0.750	&	0.702	&	0.747	&	0.748	&	0.746	&	0.709	&	0.748	&	\mi{0.748}	&	\mb{0.748}	\\
&	0.7	&	0.751	&	0.758	&	0.758	&	0.763	&	0.747	&	0.755	&	0.755	&	0.751	&	0.747	&	0.748	&	0.748	&	0.748	\\
&	0.8	&	0.757	&	0.764	&	0.764	&	0.764	&	0.755	&	0.760	&	0.761	&	0.759	&	0.748	&	0.748	&	0.748	&	0.748	\\
&	0.9	&	0.757	&	0.764	&	0.764	&	0.764	&	0.758	&	0.763	&	0.764	&	0.761	&	0.748	&	0.748	&	0.748	&	0.756	\\
\hline														
\multirow{9}{*}{T3} 
&	0.1	&	0.453	&	0.469	&	0.469	&	0.469	&	0.458	&	0.475	&	0.499'	&	0.524*	&	0.407	&	0.407	&	0.407	&	0.447	\\
&	0.2	&	0.453	&	0.594	&	\mi{0.594}	&	\mb{0.594}	&	0.545	&	0.608	&	\mi{0.615}'	&	\mbl{0.629}*	&	0.531	&	0.531	&	0.531	&	0.531	\\
&	0.3	&	0.579	&	0.679	&	\mi{0.679}	&	\mb{0.679}	&	0.616	&	0.682	&	\mi{0.698}'	&	\mb{0.701}*	&	0.577	&	0.577	&	0.577	&	0.594	\\
&	0.4	&	0.664	&	0.734	&	\mi{0.734}	&	\mb{0.734}	&	0.648	&	0.732	&	\mi{0.732}	&	\mb{0.734}	&	0.626	&	0.679	&	\mi{0.679}	&	\mb{0.688}*	\\
&	0.5	&	0.719	&	0.739	&	\mi{0.739}	&	\mb{0.739}	&	0.685	&	0.738	&	\mi{0.738}	&	\mb{0.744}	&	0.669	&	0.734	&	\mi{0.734}	&	\mb{0.734}	\\
&	0.6	&	0.737	&	0.760	&	\mi{0.760}	&	\mb{0.760}	&	0.730	&	0.752	&	\mi{0.752}	&	0.751	&	0.715	&	0.734	&	0.734	&	0.734	\\
&	0.7	&	0.758	&	0.760	&	0.760	&	0.760	&	0.750	&	0.760	&	0.760	&	0.752	&	0.676	&	0.739	&	0.739	&	0.739	\\
&	0.8	&	0.760	&	0.760	&	0.760	&	0.760	&	0.759	&	0.760	&	0.760	&	0.760	&	0.701	&	0.739	&	0.739	&	0.739	\\
&	0.9	&	0.760	&	0.760	&	0.760	&	0.760	&	0.759	&	0.760	&	0.760	&	0.760	&	0.739	&	0.749	&	0.749	&	0.754	\\
\hline														
\multirow{9}{*}{T4} 
&	0.1	&	0.343	&	0.413	&	\mi{0.413}	&	\mb{0.413}	&	0.408	&	0.423	&	\mi{0.439}'	&	\mb{0.426}	&	0.344	&	0.344	&	0.344	&	\mbl{0.363}*	\\
&	0.2	&	0.363	&	0.456	&	\mi{0.456}	&	\mb{0.456}	&	0.422	&	0.457	&	\mi{0.467}'	&	\mb{0.462}*	&	0.364	&	0.364	&	0.365	&	\mb{0.394}*		\\
&	0.3	&	0.433	&	0.491	&	\mi{0.491}	&	\mb{0.496}*	&	0.440	&	0.516	&	\mi{0.518}	&	\mb{0.508}	&	0.389	&	0.395	&	0.391	&	\mbl{0.488}*	\\
&	0.4	&	0.435	&	0.563*	&	\mil{0.563}	&	\mb{0.544}	&	0.441	&	0.540	&	\mi{0.548}'	&	\mb{0.547}*	&	0.525	&	0.556*	&	\mil{0.556}	&	0.544	\\
&	0.5	&	0.456	&	0.573	&	\mi{0.573}	&	\mbl{0.601}*	&	0.464	&	0.587	&	\mi{0.587}	&	\mb{0.588}	&	0.544	&	0.584	&	\mi{0.581}	&	\mb{0.584}	\\
&	0.6	&	0.480	&	0.608	&	\mi{0.608}	&	\mb{0.612}	&	0.503	&	0.606	&	\mi{0.606}	&	\mb{0.597}	&	0.548	&	0.609	&	\mi{0.609}	&	\mb{0.609}	\\
&	0.7	&	0.547	&	0.627*	&	\mil{0.627}	&	\mb{0.621}	&	0.539	&	0.622	&	\mi{0.622}	&	\mb{0.621}	&	0.525	&	0.609	&	\mi{0.609}	&	\mb{0.609}	\\
&	0.8	&	0.550	&	0.628	&	\mi{0.628}	&	\mb{0.627}	&	0.556	&	0.627	&	\mi{0.627}	&	\mb{0.627}	&	0.539	&	0.619	&	\mi{0.619}	&	\mb{0.619}	\\
&	0.9	&	0.555	&	0.629	&	\mi{0.629}	&	\mb{0.629}	&	0.563	&	0.629	&	\mi{0.629}	&	\mb{0.629}	&	0.559	&	0.628	&	\mi{0.628}	&	\mb{0.628}	\\
\hline														
\multirow{9}{*}{T5} 
&	0.1	&	0.376	&	0.442	&	\mi{0.442}	&	\mb{0.448}	&	0.431	&	0.471	&	\mi{0.475}'	&	\mb{0.467}	&	0.373	&	0.374	&	\mi{0.379}'	&	\mb{0.379}*	\\
&	0.2	&	0.450	&	0.516	&	\mi{0.516}	&	\mb{0.510}	&	0.458	&	0.540*	&	\mil{0.540}	&	\mb{0.528}	&	0.394	&	0.394	&	\mi{0.409}'	&	\mb{0.415}*	\\
&	0.3	&	0.501	&	0.571	&	\mi{0.571}	&	\mb{0.571}	&	0.523	&	0.577	&	\mi{0.580}'	&	\mb{0.580}*	&	0.416	&	0.545*	&	\mil{0.545}	&	\mb{0.461}	\\
&	0.4	&	0.543	&	0.608*	&	\mil{0.608}	&	\mb{0.574}	&	0.561	&	0.605	&	\mi{0.605}	&	\mb{0.600}	&	0.557	&	0.608	&	\mi{0.608}	&	\mb{0.608}	\\
&	0.5	&	0.572	&	0.621*	&	\mil{0.621}	&	\mb{0.608}	&	0.581	&	0.624	&	\mi{0.624}	&	\mb{0.617}	&	0.579	&	0.624	&	\mi{0.624}	&	\mb{0.624}	\\
&	0.6	&	0.607	&	0.638	&	\mi{0.638}	&	\mb{0.624}	&	0.608	&	0.637	&	\mi{0.637}	&	\mb{0.625}	&	0.600	&	0.626	&	\mi{0.625}	&	\mb{0.627}	\\
&	0.7	&	0.613	&	0.647	&	\mi{0.647}	&	\mb{0.647}	&	0.621	&	0.648*	&	\mil{0.648}	&	\mb{0.638}	&	0.556	&	0.648	&	\mi{0.648}	&	\mb{0.648}	\\
&	0.8	&	0.633	&	0.653	&	\mi{0.653}	&	\mb{0.651}	&	0.631	&	0.653	&	\mi{0.653}	&	\mb{0.651}	&	0.581	&	0.652	&	\mi{0.652}	&	\mb{0.652}	\\
&	0.9	&	0.642	&	0.656	&	\mi{0.656}	&	\mb{0.656}	&	0.643	&	0.655	&	\mi{0.655}	&	\mb{0.655}	&	0.569	&	0.655	&	\mi{0.655}	&	\mb{0.655}	\\
\end{tabular}}
\end{table*}

Queriability formula needs the value of the frequency 
for each concept in the input taxonomy over the dataset. 
However, it is not possible to find the exact frequencies 
of concepts without annotating the mentions to their entities 
in the dataset. 
Similar to \cite{Termehchy:SIGMOD:14}, we estimate 
the concept frequencies by sampling a small subset of 
randomly selected documents from the dataset. 
We compute the frequency of each concept using 
an estimation error rate of 5\% 
under the 95\% confidence level, 
which is about 400 documents for all datasets. 
We also smooth the sampled frequencies using 
Bayesian $m$-estimates with 
smoothing parameter of 1 and uniform priors.
We denote the \prob{Level\hyp{}wise} algorithm as {\em LW} 
and the dynamic programming algorithm as {\em DP} for brevity. 
We also compare LW and DP with the {\em APM} algorithm from 
\cite{Termehchy:SIGMOD:14} 
which finds a design over a set of concepts.
We use all concepts in the taxonomy as a set of concepts 
for an input to APM.
APM uses a scaling technique to convert popularities and costs 
to positive integers \cite{Termehchy:SIGMOD:14}.  
We set the $\epsilon$ value of the scaling for APM to 0.01. 
As we have mentioned in Section~\ref{sec:approximation-algorithms}, 
LW uses APM to find the optimal design in each level. 
We set the scaling factor of the APM algorithm used by LW to 0.01.
In addition, 
we also perform APM algorithm 
only over a set of leaf concepts in a taxonomy,
and denote this modification of APM as {\em APM-L}.

Since DP also assumes popularity ($u$), frequency ($d$) and cost ($w$)
to be positive integers, we also use scaling 
to convert the values of popularity, frequency 
and cost of all concept in the input taxonomy 
to positive integers \cite{Vazirani:Book:Approx}. 
Let $u_{\max}$ be the maximum 
popularity of all leaf concepts in the taxonomy 
and $\eps < 1$, we scale $u(C)$ as  
$\hat{u}(C) \bequal \lfloor {u(C)\over{\eps\cdot u_{\max}} }\rfloor$.
We use similar techniques
to scale the values of $d(C)$ and $w(C)$.
Intuitively, the smaller the value of $\eps$ is, 
the more exact result DP will deliver. 
However, DP takes longer to run for smaller values 
of $\eps$ as the range of $U$, $D$ and 
$B_\mathtt{total}$ will become larger.
Since we cannot use a very small value of $\epsilon$, 
our preliminary results of using this scaling technique 
shows that many concepts have $u$ or $d$ values 
equals to 0 after scaling.
Hence, DP may not explore all feasible designs 
and may miss some popular concepts whose $d$ value are small.
Thus, we modify the 
aforementioned scaling technique by
adding a constant value of 1
to both $\hat{u}(C)$ and $\hat{d}(C)$.
This addition ensures that each concept, after the scaling, 
contributes in the computation of queriability, 
and DP will explore all feasible designs. 
We set the value of $\eps$ for DP to 0.05 for the 
experiments in this section.

Table~\ref{tab:approx-mrr} 
shows the values of average MRR for APM, APM-L, LW and DP 
over T1, T2, T3, T4 and T5.
Overall, the designs returned by LW and DP improve 
the effectiveness of answering queries for all taxonomies 
more than the designs returned by APM.
This is because APM does not consider the structural information of the taxonomy.
APM often picks many popular concepts that are ancestor or descendant of each other, 
LW and DP use structural information of the taxonomy 
when selecting a design and thus avoid this problem.

Although APM-L is shown to be more effective than APM,
LW is generally as effective or more effective than APM-L.
Since APM-L returns designs only over leaf concepts of a taxonomy,
it does not have the same drawback as that of APM and
so the average MRR values of APM-L is higher than APM.
In many cases, the designs returned APM-L and LW
are equally effective.
This is because there are only two levels of concepts in T1 and T2,
and the given budgets are usually enough to select popular concepts 
in the leaf level.
Hence, both LW and APM-L returns the same designs.
In addition, although the height of T4 and T5 are higher than T1, T2 and T3,
the taxonomy trees are also unbalanced,
and for each internal concept, 
the distribution of its child concepts popularity and frequency 
are extremely skewed.
For instance, {\em organism} has only two children in T4, 
and both children are leaves.
One of the children, namely {\em person}, 
has almost the same popularity as that of {\em organism}.
Hence, a design with concepts from a leaf level 
and a design with concepts from the same level as {\em organism}
are equally effective.
Nevertheless, 
LW returns different designs than APM-L in some cases,
and they are more effective than those of APM-L.

DP is also generally more effective than APM-L.
This is because designs returned by DP can contain 
both leaf and non-leaf concepts.
For instance, given a budget of 0.3 over T2, 
APM-L picks {\em dramatic composition} and {\em literary composition}
while DP picks {\em writing} instead of {\em dramatic composition}.
Since {\em writing} is a parent of both {\em dramatic composition} 
and {\em literary composition}, 
DP can answer queries of {\em literary composition} and 
queries of other concepts that are children of {\em writing}
more effectively. 

The results shown in 
Table~\ref{tab:approx-prec3} and~\ref{tab:approx-mrr}
indicate that the designs returned by DP are more effective than
those returned by LW for small budgets and small taxonomies such as T1, T2 and T3.
However, the designs returned by LW 
are more effective than those delivered by DP 
for moderate to large budgets, e.g., 0.4-0.7.
Because the frequencies and popularities of concepts in most taxonomies
follow a power-law distribution, most of all concepts have 
frequency and popularity close to zero.
When a given budget is very small, neither methods have a sufficient budget 
to select the concepts with medium or small frequencies and popularities. 
However, over sufficiently large budgets, both algorithms are able to select some of these concepts.
If the DP algorithm does not use sufficiently small values for $\epsilon$, 
concepts with considerably different frequencies and popularities may end up with equal values of scaled frequencies and popularities.
Since it takes a very long time to run DP with an $\epsilon$ value less than 0.05, 
one cannot use a sufficiently small value of $\epsilon$ 
to preserve the differences in popularities and frequencies for all concepts
As the DP algorithm is not able to distinguish these concepts, 
it may not be able to find a design with the most queriability.
Since we are able to run LW using sufficiently small scaling factors in 
its APM component, 
LW often returns more effective designs than DP over relatively large budgets.

\begin{table*}
\caption{Average $p@3$ for APM, APM-L and LW over T6 and T8. 
Statistically significant differences between APM and LW, between APM and APM-L, and between APM-LW and LW are marked in bold, italic and underline, respectively.
$B$ denotes a given budget.}
\label{tab:approx-large-prec3}
\centering
{\scriptsize \sffamily
\begin{tabular}{r|c*{3}{|ccc}}
\multicolumn{1}{r}{} & & \multicolumn{3}{c}{Uniform Cost} & \multicolumn{3}{|c}{Random Cost} & \multicolumn{3}{|c}{Frequency-based Cost} \\
\cline{3-11}
\multicolumn{1}{r}{} & $B$ & APM & APM-L & LW & APM & APM-L & LW & APM & APM-L & LW \\
\hline\hline
\multirow{5}{*}{T6}									
&	0.1	&	0.191	&	\mi{0.221}	&	\mb{0.221}	&	0.194	&	\mi{0.219}	&	\mbl{0.220}	&	0.153	&	\mi{0.170}	&	\mb{0.170}	\\
&	0.2	&	0.229	&	\mi{0.238}	&	\mb{0.238}	&	0.230	&	\mi{0.239}	&	\mbl{0.240}	&	0.173	&	0.174	&	\mbl{0.176}	\\
&	0.3	&	0.242	&	0.247	&	0.247	&	0.241	&	\mi{0.247}	&	\mb{0.247}	&	0.176	&	0.177	&	\mbl{0.180}	\\
&	0.4	&	0.243	&	0.248	&	0.248	&	0.246	&	0.248	&	0.248	&	0.180	&	\mi{0.227}	&	\mb{0.227}	\\
&	0.5	&	0.246	&	0.248	&	0.248	&	0.247	&	0.248	&	0.248	&	0.225	&	\mi{0.238}	&	\mb{0.238}	\\
&	0.6	&	0.248	&	0.248	&	0.248	&	0.248	&	0.248	&	0.248	&	0.226	&	\mi{0.240}	&	\mb{0.241}	\\
&	0.7	&	0.248	&	0.248	&	0.248	&	0.248	&	0.248	&	0.248	&	0.226	&	\mi{0.244}	&	\mb{0.244}	\\
\hline									
\multirow{5}{*}{T7}									
&	0.1	&	0.217	&	\mi{0.233}	&	\mb{0.234}	&	0.220	&	\mi{0.235}	&	\mbl{0.237}	&	0.182	&	\mi{0.197}	&	\mb{0.197}	\\
&	0.2	&	0.244	&	\mi{0.255}	&	\mb{0.255}	&	0.244	&	\mi{0.256}	&	\mb{0.256}	&	0.199	&	0.198	&	0.198	\\
&	0.3	&	0.251	&	\mi{0.259}	&	\mb{0.259}	&	0.255	&	\mi{0.260}	&	\mb{0.260}	&	0.199	&	\mi{0.236}	&	\mb{0.236}	\\
&	0.4	&	0.260	&	0.261	&	0.261	&	0.259	&	0.261	&	0.261	&	0.237	&	\mi{0.248}	&	\mb{0.248}	\\
&	0.5	&	0.260	&	0.261	&	0.261	&	0.260	&	0.261	&	0.261	&	0.238	&	\mi{0.249}	&	\mb{0.249}	\\
&	0.6	&	0.260	&	0.261	&	0.261	&	0.260	&	0.261	&	0.261	&	0.239	&	\mi{0.251}	&	\mb{0.251}	\\
\hline
\multirow{4}{*}{T8}
&	0.1	&	0.278	&	0.279	&	0.279	&	0.279	&	0.279	&	0.279	&	0.225	&	\mi{0.243}	&	\mb{0.243}	\\
&	0.2	&	0.280	&	0.280	&	0.280	&	0.280	&	0.280	&	0.280	&	0.235	&	\mi{0.270}	&	\mb{0.270}	\\
&	0.3	&	0.280	&	0.280	&	0.280	&	0.280	&	0.280	&	0.280	&	0.260	&	\mi{0.275}	&	\mb{0.275}	\\
&	0.4	&	0.280	&	0.280	&	0.280	&	0.280	&	0.280	&	0.280	&	0.265	&	\mi{0.279}	&	\mb{0.279}	\\
&	0.5	&	0.280	&	0.280	&	0.280	&	0.280	&	0.280	&	0.280	&	0.256	&	\mi{0.279}	&	\mb{0.279}	\\
&	0.6	&	0.280	&	0.280	&	0.280	&	0.280	&	0.280	&	0.280	&	0.261	&	\mi{0.279}	&	\mb{0.279}	\\
\end{tabular}}
\end{table*}
\begin{table*}
\caption{Average MRR for APM, APM-L and LW over T6 and T8.
Statistically significant differences between APM and LW, between APM and APM-L, and between APM-LW and LW are marked in bold, italic and underline, respectively.}
\label{tab:approx-large-others}
\centering
{\scriptsize \sffamily
\begin{tabular}{r|l*{3}{|ccc}}
\multicolumn{1}{r}{} & & \multicolumn{3}{c}{Uniform Cost} & \multicolumn{3}{|c}{Random Cost} & \multicolumn{3}{|c}{Frequency-based Cost} \\
\cline{3-11}
\multicolumn{1}{r}{} & \multicolumn{1}{c|}{$B$} & APM & APM-L & LW & APM & APM-L & LW & APM & APM-L & LW \\
\hline\hline
\multirow{9}{*}{T6}
&	0.01	&	0.235	&	0.235	&	0.235	&	\mi{0.375}	&	0.291	&	\mbl{0.403}	&	0.259	&	0.259	&	0.259	\\
&	0.05	&	0.411	&	\mi{0.474}	&	\mb{0.474}	&	0.412	&	\mi{0.491}	&	\mb{0.497}	&	0.311	&	\mi{0.327}	&	\mb{0.327}	\\
&	0.1	&	0.461	&	\mi{0.564}	&	\mb{0.564}	&	0.466	&	\mi{0.559}	&	\mbl{0.563}	&	0.333	&	\mi{0.380}	&	\mb{0.380}	\\
&	0.2	&	0.583	&	\mi{0.625}	&	\mb{0.625}	&	0.589	&	\mi{0.624}	&	\mbl{0.625}	&	0.392	&	0.397	&	\mbl{0.404}	\\
&	0.3	&	0.627	&	\mi{0.648}	&	\mb{0.648}	&	0.623	&	\mi{0.647}	&	\mb{0.648}	&	0.403	&	\mi{0.410}	&	\mbl{0.417}	\\
&	0.4	&	0.631	&	\mi{0.652}	&	\mb{0.652}	&	0.638	&	\mi{0.652}	&	\mb{0.652}	&	0.418	&	\mi{0.589}	&	\mb{0.589}	\\
&	0.5	&	0.642	&	0.653	&	0.653	&	0.642	&	\mi{0.653}	&	\mb{0.653}	&	0.573	&	\mi{0.626}	&	\mb{0.626}	\\
&	0.6	&	0.646	&	0.653	&	0.653	&	0.646	&	\mi{0.653}	&	\mb{0.653}	&	0.577	&	\mi{0.632}	&	\mb{0.632}	\\
&	0.7	&	0.650	&	0.653	&	0.653	&	0.650	&	0.653	&	0.653	&	0.578	&	\mi{0.643}	&	\mb{0.643}	\\
\hline										
\multirow{9}{*}{T7}
&	0.01	&	0.344	&	\mi{0.380}	&	\mb{0.380}	&	0.385	&	0.391	&	\mbl{0.406}	&	0.286	&	0.285	& \mbl{0.310}	\\
&	0.05	&	0.448	&	\mi{0.525}	&	\mb{0.525}	&	0.451	&	\mi{0.517}	&	\mbl{0.530}	&	0.371	&	\mi{0.403}	&	\mb{0.403}	\\
&	0.1	&	0.545	&	\mi{0.609}	&	\mb{0.609}	&	0.550	&	\mi{0.612}	&	\mbl{0.616}	&	0.410	&	\mi{0.460}	&	\mb{0.460}	\\
&	0.2	&	0.637	&	\mi{0.669}	&	\mb{0.669}	&	0.633	&	\mi{0.671}	&	\mb{0.671}	&	0.464	&	0.463	&	0.461	\\
&	0.3	&	0.659	&	\mi{0.682}	&	\mb{0.682}	&	0.667	&	\mi{0.683}	&	\mb{0.683}	&	0.464	&	\mi{0.602}	&	\mb{0.602}	\\
&	0.4	&	0.677	&	0.686	&	0.686	&	0.675	&	\mi{0.686}	&	\mb{0.686}	&	0.606	&	\mi{0.642}	&	\mb{0.642}	\\
&	0.5	&	0.683	&	0.686	&	0.686	&	0.681	&	0.686	&	0.686	&	0.609	&	\mi{0.646}	&	\mb{0.646}	\\
&	0.6	&	0.684	&	0.686	&	0.686	&	0.684	&	0.686	&	0.686	&	0.611	&	\mi{0.662}	&	\mb{0.662}	\\
\hline										
\multirow{10}{*}{T8}
&	0.001	& 0.357 & 0.357 & 0.357 & \mi{0.365} & 0.334 & \mbl{0.407} & 0.277 & 0.277 & 0.277 \\
&	0.005	& 0.433	& \mi{0.452} & \mb{0.452} & 0.462 & 0.470 & \mb{0.473} & 0.316 & \mi{0.336} & \mb{0.334} \\
&	0.01	& \mbi{0.537} & 0.517 & 0.517 &	0.535 & 0.535 & 0.535 & 0.355 & \mi{0.372} & \mbl{0.379} \\
&	0.05	& 0.686	& \mi{0.707} & \mb{0.707} & 0.712 & 0.712 & 0.712 & 0.479 & \mi{0.561} & \mb{0.561}	\\
&	0.1	&	0.730	&	\mi{0.738}	&	\mb{0.738}	&	0.738	&	0.738	&	0.738	&	0.541	&	\mi{0.598}	&	\mb{0.598}	\\
&	0.2	&	0.740	&	0.743	&	0.743	&	0.742	&	0.742	&	0.742	&	0.570	&	\mi{0.704}	&	\mb{0.704}	\\
&	0.3	&	0.742	&	0.743	&	0.743	&	0.743	&	0.743	&	0.743	&	0.665	&	\mi{0.716}	&	\mb{0.716}	\\
&	0.4	&	0.743	&	0.743	&	0.743	&	0.743	&	0.743	&	0.743	&	0.680	&	\mi{0.737}	&	\mb{0.737}	\\
&	0.5	&	0.743	&	0.743	&	0.743	&	0.743	&	0.743	&	0.743	&	0.651	&	\mi{0.737}	&	\mb{0.737}	\\
&	0.6	&	0.743	&	0.743	&	0.743	&	0.743	&	0.743	&	0.743	&	0.668	&	\mi{0.737}	&	\mb{0.737}	\\
\end{tabular}}
\end{table*}

Tables~\ref{tab:approx-large-prec3} and~\ref{tab:approx-large-others} 
shows the values of average $p@3$ and MRR, respectively, 
for APM and LW over T6 and T8. 
We do not report any result for DP over T6, T7 and T8
because the algorithm does not terminate 
for almost all budgets after several days.
The results shown in Tables~\ref{tab:approx-large-prec3} 
and~\ref{tab:approx-large-others} indicate that 
the designs returned by LW are more effective 
than the ones delivered by APM. 

Because of a very skewed distribution of concept popularity
in T6, T7 and T8, 
budgets of 0.4 for T6 and T7, and 0.2 for T8  
are usually sufficient to create a design that includes 
all leaf concepts that appear in the query workloads.
Hence, APM, APM-L and LW deliver almost equally 
effective designs for relatively large budgets.
We further evaluate APM, APM-L and LW using
budgets 0.01 and 0.05 for T6, T7 and T8
and budgets 0.001 and 0.005 for T8
as shown in Table~\ref{tab:approx-large-others}.
Overall, LW is significantly more effective than APM-L and APM.
This is because, with small budget, it is more preferable 
to choose a concept in a higher level instead of multiple
leaf concepts that are children or descendants of that one concept. 
Then an algorithm can spend the remaining budget on other concepts
to help answer other queries more effectively.

\subsubsection{Dynamic Programming with Cost-Dependency}

\begin{table}
\caption{Average $p@3$ of QM and DPC 
using different values of $\epsilon$ over T1, T2 and T3.
Statistically significant difference between QM and DPC 
are marked in bold. $B$ denotes a given budget.}
\label{tab:costdep-prec3}
\centering
{\scriptsize \sffamily
\begin{tabular}{c|c|c|ccc}
\multicolumn{1}{c}{} & $B$ & QM & DPC$_{0.05}$ & DPC$_{0.1}$ & DPC$_{0.2}$ \\
\hline\hline
\multirow{4}{*}{T1}
& 0.25 &	0.170	& \mb{0.127} &	\mb{0.149}	&	\mb{0.165}	\\
& 0.50  &	0.189	& 0.187 &	0.187	&	0.187	\\
& 0.75 &	0.195	& 0.195 &	0.195	&	0.189	\\
& 1    &	0.195	& 0.195 &	0.195	&	0.195	\\
\hline
\multirow{4}{*}{T2}
& 0.25 &	0.285	& \mb{0.270} &	\mb{0.269}	&	\mb{0.269}	\\
& 0.50  &	0.321	& 0.319 &	0.318	&	\mb{0.300}	\\
& 0.75 &	0.325	& 0.323 &	0.322	&	0.321	\\
& 1    &	0.326	& 0.326 &	0.326	&	0.326	\\
\hline
\multirow{4}{*}{T3}
& 0.25 &	0.285	& \mb{0.268} &	\mb{0.268}	&	\mb{0.268}	\\
& 0.50  &	0.305	& 0.304 &	0.304	&	\mb{0.292}	\\
& 0.75 &	0.306	& 0.305 &	0.304	&	0.304	\\
& 1    &	0.306	& 0.306 &	0.306	&	0.306	\\
\end{tabular}}
\end{table}
\begin{table}
\caption{Average MRR of QM and DPC 
using different values of $\epsilon$ over T1, T2 and T3.
Statistically significant difference between QM and DPC 
are marked in bold. $B$ denotes a given budget.}
\label{tab:costdep-others}
\centering
{\scriptsize \sffamily
\begin{tabular}{c|c|c|ccc}
\multicolumn{1}{c}{} & $B$ & QM & DPC$_{0.05}$ & DPC$_{0.1}$ & DPC$_{0.2}$ \\
\hline\hline
\multirow{4}{*}{T1}
& 0.25 &	0.426	& \mb{0.299} &	\mb{0.364}	&	0.413	\\
& 0.50  &	0.488	& 0.464 &	0.464	&	0.464	\\
& 0.75 &	0.507	& 0.507 &	0.507	&	0.507	\\
& 1    &	0.507	& 0.507 &	0.507	&	0.507	\\
\hline
\multirow{4}{*}{T2}
& 0.25 &	0.609	& 0.596	&	0.594	&	0.594	\\
& 0.50  &	0.747	& 0.743	&	0.742	&	\mb{0.695}	\\
& 0.75 &	0.763	& 0.759 &	0.755	&	0.750	\\
& 1    &	0.764	& 0.764	&	0.764	&	0.764	\\
\hline
\multirow{4}{*}{T3}
& 0.25 &	0.707	& 0.662 &	\mb{0.655}	&	\mb{0.655}	\\
& 0.50  &	0.756	& 0.747 &	0.741	&	\mb{0.729}	\\
& 0.75 &	0.760	& 0.760 &	0.758	&	0.749	\\
& 1    &	0.760	& 0.760 &	0.760	&	0.760	\\
\end{tabular}}
\end{table}

Tables~\ref{tab:costdep-prec3} and~\ref{tab:costdep-others}
show the values of $p@3$ and MRR, respectively, for 
Queriability Maximization (QM) and dynamic programming algorithm 
with cost dependencies ({\it DPC}) over T1, T2 and T3 taxonomies. 
The cost for each concept in these 
taxonomies has been randomly generated and depends on 
which ancestors of the concept have been selected in the 
design. The budgets that cover all leaf nodes 
in T1, T2 and T3 are 1, and the budget that cover all nodes
are 1.25, 1.1 and 1.1, respectively.
QM is a brute force algorithm that explores all feasible
solutions and finds the design with maximum queriability.
Because the space of the possible solutions for this problem is
larger than the original \prob{CECD} problem, it takes
much longer to run QM for this problem. 
Hence, we have covered a smaller range of budgets in this set of experiments.
The results shown in Tables~\ref{tab:costdep-prec3} 
and~\ref{tab:costdep-others} indicate that 
the smaller the value of $\epsilon$, the closer the average MRR
of the designs returned by DPC are to the ones delivered by QM.

\ignore{
Generally, the smaller the value of epsilon, the closer average $p@3$
of the designs returned by DPC to the ones delivered by QM, 
According to MRR results shown in Table~\ref{tab:costdep-others}, 
QM is significantly better than DPC in many cases.
This is because, the relevant answers in the ranked list of answers returned
by the query interface that uses the designs selected by DPC are placed
within top 3, but they are usually at the second or third positions in
the list.
However, these relevant answers are usually found at the first or second positions
in the ranked list of answers returned by the query interface that uses
the designs returned by QM. 
The difference in average values of MRR of the ranked list of answers
from using the design returned by QM and DPC can be as large as 0.5.
However, there is no difference when using $p@3$ as an effectiveness measurement.
} 

\subsection{Efficiency of Proposed Algorithms}
\label{sec:efficiency}

We measure the running times of LW and DP  
over moderate and large taxonomies, 
i.e., T4, T5, T6, T7 and T8, and set the 
available main memory of Java Virtual Machine to 64GB.
Table~\ref{tab:runtime} shows the average running times of 
APM, LW and DP for T4, T5, T6, T7 and T8 over budgets 0.1 to 0.9. 
Some results of DP are not reported because 
the algorithms did not finish after a day.
Overall, LW is as efficient as APM, and it is more
efficient than DP.
Because the size of the table required in the DP algorithm is substantially 
large for $\epsilon = 0.05$, it occupies most of the available main memory. 
Thus, the running time of DP is longer than APM and LW on average.
Therefore, LW scales for large taxonomy and is efficient for a design-time task.
On the other hand, DP has a reasonable running time for T4 and T5,
but it does not scale for large taxonomy such as T6, T7 and T8.

\begin{table}
\caption{Average running time of APM, LW, DP and DPC 
where $m$ and $s$ denote minute and second.}
\label{tab:runtime}
\centering
{\scriptsize \sffamily
\begin{tabular}{l|c|c|rr|rr}
& \multirow{2}{*}{APM} & \multirow{2}{*}{LW} & \multicolumn{2}{c|}{DP} & \multicolumn{2}{c}{DPC} \\
\cline{4-7}
& & & $\epsilon =$ 0.05 & 0.1 & $\epsilon =$ 0.1 & 0.2 \\
\hline\hline
T4 & \multicolumn{1}{r|}{1$s$} & \multicolumn{1}{r|}{1$s$}	& 127$m$ & 11$m$ & 151$m$ & 12$m$ \\ 
T5 & \multicolumn{1}{r|}{1$s$} & \multicolumn{1}{r|}{1$s$}	& 184$m$ & 15$m$ & 778$m$ & 77$m$ \\ 
\hline
T6 & \multicolumn{1}{r|}{2$s$} & \multicolumn{1}{r|}{1$s$}	& -	& - & - & 520$m$ 	\\ 
T7 & \multicolumn{1}{r|}{3$s$} & \multicolumn{1}{r|}{3$s$}	& -	& - & - & - 	\\ 
T8	& \multicolumn{1}{r|}{5$s$}	 & \multicolumn{1}{r|}{9$s$}		& - & -	& - & - 	\\
\end{tabular}}
\end{table}

\subsubsection{Dynamic Programming with Cost-Dependency}

Table~\ref{tab:runtime} shows the average running times 
of DPC for T4, T5, T6, T7 and T8 over budgets 0.25, 0.5, 0.75 and 1 
using the scaling factor, $\epsilon$, of 0.1 and 0.2. 
We do not report the running time of DPC 
using $\epsilon \bequal 0.1$ for T6, T7 and T8
and DPC using $\epsilon \bequal 0.2$ for T7 and T8 
because the algorithm did not finish after a day.
Because DPC requires a larger table than the one required for DP,
the running time of DPC is longer than that of DP 
for the same scaling factor, i.e., $\epsilon$.
Hence, we run DPC using only $\epsilon$ values of 0.1 and 0.2.
Overall, DPC with $\epsilon \bequal 0.1$ is reasonably efficient 
to perform a design-time task for a taxonomy of up to size 70, 
and DPC with $\epsilon \bequal 0.2$ is reasonably efficient  for a taxonomy of up to size 200.

\subsection{Queries With Multiple Concepts}
\label{sec:me-cecd-exp}

\begin{table}
\caption{
The numbers of queries with multiple concepts with 
the minimum, average and maximum numbers of concepts per query 
for  T4, T5, T6, T7 and T8.
}
\label{tab:expinfo2}
\centering 
{\sffamily
\begin{tabular}{r|c|c|c|c|c}
Taxonomy & T4 & T5 & T6 & T7 & T8 \\
\hline
\#queries	& 687 & 1578 & 1403 & 1882 & 2603 \\
minimum	\#concepts & 2 & 2 & 2 & 2 & 2 \\
maximum	\#concepts & 4 & 4 & 5 & 5 & 5 \\
average	\#concepts & 2.2 & 2.1 & 2.2 & 2.2 & 2.2 \\
\end{tabular}}
\end{table}

\subsubsection{Validation and Effectiveness}

We have selected all queries with 
multiple concepts which belong to T1, T2 and T3 
from our query workload and filtered out the queries
whose ranking quality is not improved by annotating
all concepts in the corresponding taxonomy.
This results in 6, 37 and 8 queries over 
T1, T2 and T3, respectively.
The number of concepts for each query is 2.
Because there are not enough queries with multiple concepts 
for T1 and T3, we do not evaluate our models over T1 and T3.
Since the number of queries with one concept over T2 
is seven times larger than the number of multiple concepts, 
we randomly select a subset of queries with one concept 
from the original query workload
and combine them with the queries with multiple concepts over T2. 
The new query workload contains 106 queries.
We run Oracle, 
the queriability maximization over queries with multiple concepts ($MQM$), 
the \prob{Level\hyp{}wise} algorithm for multiple concepts ($MLW$), 
and the \prob{Level\hyp{}wise} algorithm ($LW$) over the query workload. 
Oracle is described in Section~\ref{sec:expvalidate}.
MQM enumerates all feasible designs over the input taxonomy 
and returns the one with maximum queriability 
as computed in Section~\ref{sec:me-appendix}.
Table~\ref{tab:mult-others} shows the values of average MRR
for Oracle, MQM, MLW and LW algorithms over T2.
MQM generally returns the designs similar to those of Oracle. 
The designs returned by MLW also deliver close average MRR 
the ones delivered by MQM in most cases.
Overall, MLW significantly outperforms LW over T2.

\begin{table}
\caption{Average $p@3$ for Oracle, MQM, MLW and LW over T2.
Statistically significant difference between MQM and Oracle, MQM and MLW, and MLW and LW 
are marked in italic, bold, and underline, respectively. $B$ denotes a given budget.}
\label{tab:mult-prec3}
\centering
{\scriptsize \sffamily
\begin{tabular}{@{\hspace{0.5em}}c|c@{\hspace{0.7em}}c@{\hspace{0.7em}}c@{\hspace{0.7em}}c|c@{\hspace{0.7em}}c@{\hspace{0.7em}}c@{\hspace{0.7em}}c|c@{\hspace{0.7em}}c@{\hspace{0.7em}}c@{\hspace{0.7em}}c}
& \multicolumn{4}{c}{Uniform Cost} & \multicolumn{4}{|c}{Random Cost} & \multicolumn{4}{|c}{Frequency-based Cost} \\
\cline{2-13}
$B$ & Oracle & MQM & MLW & LW & Oracle & MQM & MLW & LW & Oracle & MQM & MLW & LW \\
\hline\hline
0.1	&	0.283	&	0.283	&	0.230	&	0.230	&	0.288	&	\mb{0.285}	&	0.225	&	0.216	& & 0.230 & 0.230 & 0.230 \\
0.2	&	\mi{0.311}	&	\mb{0.286}	&	0.233	&	0.223	&	0.311	&	\mb{0.291}	&	\ml{0.246}	&	0.231	& & 0.311 & 0.311 & 0.311 \\
0.3	&	\mi{0.318}	&	\mb{0.296}	&	\ml{0.248}	&	0.226	&	\mi{0.321}	&	0.300	&	\ml{0.282}	&	0.242	& & 0.320 & 0.320 & 0.320 \\
0.4	&	0.327	&	0.311	&	\ml{0.311}	&	0.230	&	0.327	&	0.313	&	\ml{0.314}	&	0.260	& & 0.324 & 0.324 & 0.324\\
0.5	&	0.330	&	0.321	&	\ml{0.321}	&	0.239	&	0.329	&	0.325	&	\ml{0.325}	&	0.266	& & 0.327 & 0.327 & 0.327 \\
0.6	&	0.330	&	0.327	&	\ml{0.327}	&	0.239	&	0.330	&	0.327	&	\ml{0.325}	&	0.270	& & 0.327 & 0.327 & 0.327 \\
0.7	&	0.330	&	0.330	&	0.330	&	0.324	&	0.330	&	0.328	&	0.328	&	0.312	& & 0.327 & 0.327 & 0.327 \\
0.8	&	0.330	&	0.330	&	0.330	&	0.327	&	0.330	&	0.330	&	0.330	&	0.326	& & 0.327 & 0.327 & 0.327\\
0.9	&	0.330	&	0.330	&	0.330	&	0.330	&	0.330	&	0.330	&	0.330	&	0.330	& & 0.327 & 0.327 & 0.327 \\
\end{tabular}}
\end{table}
\begin{table*}
\caption{Average MRR for Oracle, MQM, MLW and LW over T2. 
Statistically significant difference between MQM and Oracle, MQM and MLW, and MLW and LW 
are marked in italic, bold, and underline, respectively. $B$ denotes a given budget.}
\label{tab:mult-others}
\centering
{\scriptsize \sffamily
\begin{tabular}{@{\hspace{0.5em}}c|c@{\hspace{0.7em}}c@{\hspace{0.7em}}c@{\hspace{0.7em}}c|c@{\hspace{0.7em}}c@{\hspace{0.7em}}c@{\hspace{0.7em}}c|c@{\hspace{0.7em}}c@{\hspace{0.7em}}c@{\hspace{0.7em}}c}
& \multicolumn{4}{c}{Uniform Cost} & \multicolumn{4}{|c}{Random Cost} & \multicolumn{4}{|c}{Frequency-based Cost} \\
\cline{2-13}
$B$ & Oracle & MQM & MLW & LW & Oracle & MQM & MLW & LW & Oracle & MQM & MLW & LW \\
\hline\hline
0.1	&	\mi{0.517}	&	0.430	&	0.491	&	0.491	&	\mi{0.577}	&	0.491	&	0.465	&	0.454	&	0.569	&	0.569	&	0.569	&	0.569	\\
0.2	&	0.593	&	\mb{0.577}	&	0.516	&	0.523	&	0.651	&	0.596	&	\ml{0.569}	&	0.529	&	0.657	&	0.657	&	0.657	&	0.657	\\
0.3	&	\mi{0.747}	&	0.602	&	\ml{0.600}	&	0.526	&	\mi{0.785}	&	0.657	&	\ml{0.670}	&	0.577	&	\mi{0.758}	&	0.682	&	0.682	&	0.682	\\
0.4	&	0.826	&	0.796	&	\ml{0.796}	&	0.563	&	0.826	&	0.791	&	\ml{0.791}	&	0.618	&	0.829	&	0.829	&	0.829	&	0.829	\\
0.5	&	0.850	&	0.821	&	\ml{0.821}	&	0.576	&	0.846	&	0.832	&	\ml{0.832}	&	0.626	&	0.832	&	0.832	&	0.832	&	0.829	\\
0.6	&	0.859	&	0.852	&	\ml{0.852}	&	0.576	&	0.862	&	0.837	&	\ml{0.837}	&	0.663	&	0.832	&	0.832	&	0.832	&	0.832	\\
0.7	&	0.865	&	0.865	&	0.865	&	0.823	&	0.862	&	0.852	&	\ml{0.852}	&	0.768	&	0.832	&	0.832	&	0.832	&	0.832	\\
0.8	&	0.865	&	0.865	&	0.865	&	0.832	&	0.862	&	0.862	&	0.862	&	0.835	&	0.832	&	0.832	&	0.832	&	0.832	\\
0.9	&	0.865	&	0.865	&	0.865	&	0.865	&	0.862	&	0.862	&	0.862	&	0.858	&	0.861	&	0.861	&	\ml{0.861}	&	0.832	\\
\end{tabular}}
\end{table*}

We have also evaluated the effectiveness of MLW and LW 
over taxonomies T4, T5, T6, T7 and T8.
We have selected all queries with multiple concepts 
over T4, T5, T6, T7 and T8 from our query workload.
Table~\ref{tab:expinfo2} shows the numbers of queries 
with multiple concepts and 
the minimum, average and maximum numbers of
concepts per query for T4, T5, T6, T7 and T8.
Tables~\ref{tab:mul-large-prec3} and~\ref{tab:mul-large-others} 
shows the values of $p@3$ and MRR, respectively, 
for MLW and LW algorithms over T4, T5, T6, T7 and T8.
Overall, the designs returned by MLW and $p@3$ deliver significantly higher
average MRR and $p@3$ than the designs selected by LW over all taxonomies.

We have observed less difference between the results of MLW and LW when the distribution 
of concept popularity in the taxonomy is less skewed.
This is because, when the popularity distribution is very skewed, 
both algorithms select all or most relatively popular concepts.
Hence, they selected designs are almost equally effective.
For example, the distribution of concept frequencies in T6 
is considerably more skewed than those of the concepts in T7, 
thus, the designs delivered by MLW and LW over T6 are significantly more different 
than the ones these algorithms return over T7.

\begin{table}
\caption{Average $p@3$ for MLW and LW over T4, T5, T6, T7 and T8. 
Statistically significant differences between MLW and LW are marked in bold.
$B$ denotes a given budget.}
\label{tab:mul-large-prec3}
\centering
{\scriptsize \sffamily
\begin{tabular}{r|c*{3}{|cc}}
\multicolumn{1}{r}{} & & \multicolumn{2}{c}{Uniform Cost} & \multicolumn{2}{|c}{Random Cost} & \multicolumn{2}{|c}{Frequency-based Cost} \\
\cline{3-8}
\multicolumn{1}{r}{} & $B$ & MLW & LW & MLW & LW & MLW & LW \\
\hline\hline
\multirow{9}{*}{T4}									
&	0.1	&	\mb{0.169}	&	0.134	&	\mb{0.201}	&	0.168	& \mb{0.225} & 0.206 \\
&	0.2	&	\mb{0.248}	&	0.216	&	\mb{0.253}	&	0.210	& \mb{0.272} & 0.240 \\
&	0.3	&	\mb{0.286}	&	0.220	&	\mb{0.307}	&	0.247	& 0.282	& 0.282 \\
&	0.4	&	\mb{0.319}	&	0.279	&	\mb{0.320}	&	0.272	& 0.254 & \mb{0.285} \\
&	0.5	&	\mb{0.319}	&	0.283	&	\mb{0.322}	&	0.287	& 0.296 & 0.289 \\
&	0.6	&	0.319	&	0.306	&	0.322	&	0.303	& 0.310 & 0.310 \\
&	0.7	&	0.326	&	0.315	&	0.326	&	0.315	& 0.310 & 0.315 \\
&	0.8	&	0.326	&	0.315	&	0.327	&	0.321	& 0.322	& 0.317 \\
&	0.9	&	0.329	&	0.329	&	0.328	&	0.329	& 0.329 & 0.317 \\
\hline									
\multirow{9}{*}{T5}									
&	0.1	&	\mb{0.200}	&	0.175	&	\mb{0.253}	&	0.186	& \mb{0.269} & 0.230 \\
&	0.2	&	\mb{0.299}	&	0.257	&	\mb{0.300}	&	0.250	& 0.279 & 0.277 \\
&	0.3	&	0.301	&	0.301	&	\mb{0.301}	&	0.276	& 0.281 & 0.286 \\
&	0.4	&	0.302	&	0.279	&	0.302	&	0.293	& \mb{0.302} & 0.289 \\
&	0.5	&	0.303	&	0.302	&	0.303	&	0.303	& \mb{0.304} & 0.294 \\
&	0.6	&	0.306	&	0.306	&	0.311	&	0.306	& 0.309 & 0.296 \\
&	0.7	&	0.315	&	0.305	&	0.315	&	0.306	& 0.310 & 0.307 \\
&	0.8	&	0.315	&	0.308	&	0.315	&	0.308	& 0.313 & 0.309 \\
&	0.9	&	0.315	&	0.314	&	0.315	&	0.312	& 0.315 & 0.313 \\
\hline									
\multirow{9}{*}{T6}									
&	0.1	&	\mb{0.288}	&	0.186	&	\mb{0.289}	&	0.207	& \mb{0.283} & 0.164 \\
&	0.2	&	\mb{0.295}	&	0.203	&	\mb{0.303}	&	0.217	& \mb{0.288} & 0.203 \\
&	0.3	&	\mb{0.316}	&	0.249	&	\mb{0.316}	&	0.256	& \mb{0.301} & 0.206 \\
&	0.4	&	\mb{0.317}	&	0.287	&	0.316	&	0.292	& \mb{0.309} & 0.275 \\
&	0.5	&	0.317	&	0.321	&	0.317	&	0.321	& \mb{0.314} & 0.299 \\
&	0.6	&	0.317	&	\mb{0.327}	&	0.320	&	0.321	& 0.321 & 0.322 \\
&	0.7	&	0.327	&	0.327	&	0.328	&	0.304	& 0.325 & 0.324 \\
&	0.8	&	0.328	&	0.327	&	0.328	&	0.326	& 0.326 & 0.324 \\
&	0.9	&	0.328	&	0.327	&	0.328	&	0.326	& 0.328 & 0.326 \\
\hline									
\multirow{9}{*}{T7}									
&	0.1	&	0.285	&	0.291	&	0.286	&	0.291	& \mb{0.244} & 0.160 \\
&	0.2	&	\mb{0.307}	&	0.295	&	\mb{0.306}	&	0.211	& \mb{0.276} & 0.249 \\
&	0.3	&	0.307	&	0.305	&	0.307	&	0.306	& 0.282 & 0.284 \\
&	0.4	&	0.309	&	0.312	&	0.309	&	0.312	& \mb{0.291} & 0.284 \\
&	0.5	&	\mb{0.313}	&	0.286	&	\mb{0.313}	&	0.285	& \mb{0.305} & 0.296 \\
&	0.6	&	0.313	&	0.309	&	0.314	&	0.307	& 0.308 & 0.308 \\
&	0.7	&	0.313	&	0.318	&	0.313	&	0.318	& 0.310 & 0.310 \\
&	0.8	&	0.313	&	0.318	&	0.313	&	0.318	& 0.313 & 0.313 \\
&	0.9	&	0.318	&	0.318	&	0.318	&	0.318	& 0.317 & 0.317 \\
\end{tabular}}
\end{table}
\begin{table}
\caption{Average MRR for MLW and LW over T4, T5, T6, T7 and T8. 
Statistically significant differences between MLW and LW are marked in bold.
$B$ denotes a given budget.}
\label{tab:mul-large-others}
\centering
{\scriptsize \sffamily
\begin{tabular}{r|c*{2}{|cc}|ccc}
\multicolumn{1}{r}{} & & \multicolumn{2}{c}{Uniform Cost} & \multicolumn{2}{|c}{Random Cost} & \multicolumn{3}{|c}{Frequency-based Cost} \\
\cline{3-9}
\multicolumn{1}{r}{} & $B$ & MLW & LW & MLW & LW & MLW & LW \\
\hline\hline
\multirow{9}{*}{T4}									
&	0.1	&	\mb{0.334}	&	0.274	&	\mb{0.448}	&	0.368	&	\mb{0.574}	&	0.487	\\
&	0.2	&	\mb{0.622}	&	0.488	&	\mb{0.638}	&	0.487	&	\mb{0.721}	&	0.629	\\
&	0.3	&	\mb{0.748}	&	0.509	&	\mb{0.827}	&	0.601	&	0.770	&	0.770	\\
&	0.4	&	\mb{0.873}	&	0.712	&	\mb{0.880}	&	0.712	&	0.773	&	0.784	\\
&	0.5	&	\mb{0.873}	&	0.733	&	\mb{0.887}	&	0.749	&	0.815	&	0.799	\\
&	0.6	&	0.873	&	0.833	&	\mb{0.890}	&	0.813	&	0.845	&	0.799	\\
&	0.7	&	0.901	&	0.872	&	0.901	&	0.868	&	0.845	&	0.862	\\
&	0.8	&	0.901	&	0.862	&	0.912	&	0.894	&	0.902	&	0.872	\\
&	0.9	&	0.929	&	0.930	&	0.924	&	0.930	&	0.900	&	0.873	\\
\hline									
\multirow{9}{*}{T5}									
&	0.1	&	\mb{0.466}	&	0.430	&	\mb{0.619}	&	0.466	&	\mb{0.674}	&	0.536	\\
&	0.2	&	\mb{0.799}	&	0.673	&	\mb{0.803}	&	0.645	&	0.708	&	0.695	\\
&	0.3	&	0.807	&	0.810	&	\mb{0.819}	&	0.729	&	0.719	&	0.728	\\
&	0.4	&	\mb{0.826}	&	0.743	&	\mb{0.828}	&	0.784	&	\mb{0.827}	&	0.745	\\
&	0.5	&	0.834	&	0.811	&	0.833	&	0.819	&	\mb{0.838}	&	0.773	\\
&	0.6	&	0.848	&	0.848	&	0.866	&	0.845	&	\mb{0.852}	&	0.800	\\
&	0.7	&	0.883	&	0.840	&	0.883	&	0.840	&	0.863	&	0.851	\\
&	0.8	&	0.883	&	0.846	&	0.883	&	0.850	&	0.874	&	0.856	\\
&	0.9	&	0.883	&	0.869	&	0.883	&	0.862	&	0.883	&	0.876	\\
\hline									
\multirow{9}{*}{T6}										
&	0.1	&	\mb{0.794}	&	0.473	&	\mb{0.796}	&	0.545	&	0.765	&	0.763	\\
&	0.2	&	\mb{0.813}	&	0.508	&	\mb{0.840}	&	0.561	&	\mb{0.827}	&	0.777	\\
&	0.3	&	\mb{0.883}	&	0.667	&	\mb{0.882}	&	0.690	&	\mb{0.853}	&	0.799	\\
&	0.4	&	\mb{0.884}	&	0.795	&	\mb{0.884}	&	0.810	&	\mb{0.857}	&	0.813	\\
&	0.5	&	0.890	&	0.887	&	0.890	&	0.887	&	\mb{0.871}	&	0.820	\\
&	0.6	&	0.892	&	\mb{0.920}	&	0.894	&	0.897	&	0.888	&	0.887	\\
&	0.7	&	0.901	&	0.838	&	\mb{0.912}	&	0.833	&	\mb{0.908}	&	0.896	\\
&	0.8	&	0.932	&	0.906	&	0.932	&	0.913	&	\mb{0.932}	&	0.903	\\
&	0.9	&	0.932	&	0.906	&	0.932	&	0.913	&	0.939	&	0.926	\\
\hline										
\multirow{9}{*}{T7}										
&	0.1	&	0.764	&	0.778	&	0.771	&	0.778	&	\mb{0.723}	&	0.605	\\
&	0.2	&	\mb{0.858}	&	0.805	&	\mb{0.856}	&	0.552	&	\mb{0.762}	&	0.717	\\
&	0.3	&	0.860	&	0.848	&	0.860	&	0.852	&	\mb{0.779}	&	0.731	\\
&	0.4	&	0.872	&	0.876	&	0.872	&	0.876	&	\mb{0.826}	&	0.791	\\
&	0.5	&	\mb{0.875}	&	0.768	&	\mb{0.876}	&	0.767	&	0.845	&	0.846	\\
&	0.6	&	0.880	&	0.867	&	0.882	&	0.861	&	0.855	&	0.855	\\
&	0.7	&	0.889	&	0.895	&	0.889	&	0.895	&	0.866	&	0.866	\\
&	0.8	&	0.889	&	0.895	&	0.889	&	0.845	&	0.875	&	0.874	\\
&	0.9	&	0.890	&	0.896	&	0.901	&	0.896	&	0.895	&	0.895	\\
\hline
\multirow{9}{*}{T8}	
&	0.1	&	\mb{0.882}	&	0.783	&	\mb{0.880}	&	0.793	&	\mb{0.776}	&	0.548	\\
&	0.2	&	\mb{0.883}	&	0.848	&	\mb{0.890}	&	0.844	&	\mb{0.785}	&	0.511	\\
&	0.3	&	\mb{0.887}	&	0.848	&	\mb{0.893}	&	0.859	&	\mb{0.791}	&	0.589	\\
&	0.4	&	\mb{0.889}	&	0.867	&	0.895	&	0.876	&	\mb{0.798}	&	0.654	\\
&	0.5	&	\mb{0.894}	&	0.887	&	0.896	&	0.888	&	\mb{0.802}	&	0.736	\\
&	0.6	&	0.896	&	0.896	&	0.898	&	0.892	&	\mb{0.846}	&	0.750	\\
&	0.7	&	0.896	&	0.896	&	0.900	&	0.892	&	\mb{0.881}	&	0.754	\\
&	0.8	&	0.901	&	0.901	&	0.901	&	0.892	&	0.891	&	0.894	\\
&	0.9	&	0.901	&	0.901	&	0.901	&	0.894	&	0.900	&	0.900	\\
\end{tabular}}
\end{table}

\subsubsection{Efficiency}

We measure the average running times of MLW over moderate and large taxonomies, 
i.e., T4, T5, T6, T7 and T8,
and set the available main memory of the Java Virtual Machine to 64GB.
The average running times of MLW for T4, T5, T6, T7 and T8 
over budgets 0.1 to 0.9 are 3, 3, 4, 4 and 8 minutes, respectively,
which are reasonable for a design-time task.

\section{Conclusion and Future Work} 
\label{section:conclusion}
\label{sec:conclusion}
Annotating entities in large unstructured or
semi-structured data sets improves the effectiveness
of answering queries over these data sets.
It takes significant amounts of 
financial and computational resources 
and/or manual labor to annotate entities 
of a concept. Because an enterprise normally has limited 
resources, it has to choose a subset of
affordable concepts in its domain of interest for annotation. 
In this paper, we introduced 
the problem of cost-effective conceptual design using taxonomies, 
where given a taxonomy, one would like to find a 
subset of concepts in the taxonomy whose
total cost does not exceed a given budget and improves the 
effectiveness of answering queries the most.
We proved the problem is NP-hard and proposed an efficient approximation algorithm, called \prob{level-wise} 
algorithm, and an exact algorithm with 
pseudo-polynomial running time for the problem
over tree taxonomies. 
We also proved that it is not possible 
to find any approximation algorithm with reasonably 
small approximation ratio or pseudo-polynomial 
time exact algorithm for the problem
when the taxonomy is a directed acyclic graph.
We showed that our formalization framework effectively estimates 
the amount by which a design improves the effectiveness of 
answering queries through extensive experiments over real-world 
datasets, taxonomies, and queries. Our empirical studies 
also indicated that our algorithms are efficient 
for a design-time task
with pseudo-polynomial algorithm delivering more effective 
designs in most cases.

\bibliographystyle{abbrv}
\bibliography{ref}

\end{document}